\documentclass[twocolumn]{svjour3}          % twocolumn
\smartqed  % flush right qed marks, e.g. at end of proof
\usepackage{amsmath}
\usepackage{amssymb}
\usepackage{latexsym}
\usepackage{graphicx}
\usepackage{subfigure}
\usepackage{multirow}
\usepackage[ruled, vlined, linesnumbered]{algorithm2e}
\usepackage{url}
\usepackage{enumitem}
\usepackage[dvipsnames]{xcolor}
\usepackage{lipsum}
\usepackage{epstopdf}
\usepackage{balance}
\usepackage{bbding}

\usepackage[normalem]{ulem}
\newcommand{\rr}[1]{\textcolor{black}{#1}}

\newtheorem{observation}{Observation}
\newcommand{\kw}[1]{{\ensuremath {\mathsf{#1}}}\xspace}

\newcommand{\reffig}[1]{Figure~\ref{fig:#1}}
\newcommand{\refsec}[1]{Section~\ref{sec:#1}}
\newcommand{\refsubsec}[1]{Section~\ref{subsec:#1}}
\newcommand{\reftab}[1]{Table~\ref{tab:#1}}
\newcommand{\refalg}[1]{Algorithm~\ref{alg:#1}}

\newcommand{\refdef}[1]{Definition~\ref{def:#1}}
\newcommand{\refthm}[1]{Theorem~\ref{thm:#1}}
\newcommand{\reflem}[1]{Lemma~\ref{lem:#1}}

\newcommand{\stitle}[1]{\noindent{\bf #1}}

\newcommand{\nbr}{N}

\newcommand{\avgdeg}{d}

\newcommand{\height}{h}

\newcommand{\dtree}{D-Tree\xspace}
\newcommand{\hdt}{HDT\xspace}
\newcommand{\idtdsa}{ID-TDSA\xspace}
\newcommand{\idtdsb}{ID-TDSB\xspace}

\newcommand{\parent}{parent}
\newcommand{\depth}{depth}
\newcommand{\stsize}{st\_size}

\newcommand{\kwnull}{\kw{Null}}

\newcommand{\basename}{ID}
\newcommand{\baseidxname}{\basename-Tree\xspace}
\newcommand{\basetree}{\basename-Tree\xspace}
\newcommand{\basetrees}{\basename-Trees\xspace}

\newcommand{\basetname}{ID$^{2}$}
\newcommand{\basetidxname}{\basetname Tree\xspace}

\newcommand{\opttname}{DND$^{2}$}
\newcommand{\opttidxname}{\opttname Tree\xspace}

\newcommand{\algbasequery}{\kw{\basename\text{-}Query}}
\newcommand{\algbaseinsert}{\kw{\basename\text{-}Insert}}
\newcommand{\algbasedelete}{\kw{\basename\text{-}Delete}}

\newcommand{\alglink}{\kw{Link}}
\newcommand{\algunlink}{\kw{Unlink}}
\newcommand{\algreroot}{\kw{ReRoot}}

\newcommand{\algtdeletent}{\kw{\basename 2DeleteNonTreeEdge}}
\newcommand{\algtdeletet}{\kw{\basename 2DeleteTreeEdge}}

\newcommand{\algtddeletent}{\kw{\optname 2DeleteNonTreeEdge}}
\newcommand{\algtddeletet}{\kw{\optname 2DeleteTreeEdge}}

\newcommand{\algtquery}{\kw{\basename 2Query}}

\newcommand{\algteinsert}{\kw{\basename 2Insert}}
\newcommand{\algtedinsert}{\kw{\optname 2Insert}}

\newcommand{\algcut}{\ensuremath{\kw{Cut\text{-}Bridge}}\xspace}

\newcommand{\dsname}{DS}
\newcommand{\dstree}{\dsname-Tree\xspace}
\newcommand{\dstrees}{\dsname-Trees\xspace}
\newcommand{\optname}{DND}
\newcommand{\optidxname}{\optname-Trees\xspace}
\newcommand{\algoptinsert}{\kw{\optname\text{-}Insert}}
\newcommand{\algoptdelete}{\kw{\optname\text{-}Delete}}
\newcommand{\algoptquery}{\kw{\optname\text{-}Query}}

\newcommand{\algdsdel}{\kw{Isolate}}

\newcommand{\algdsreroot}{\kw{ReRootDS}}
\newcommand{\algdsunlink}{\kw{UnlinkDS}}
\newcommand{\algdslink}{\kw{LinkDS}}

\newcommand{\algfind}{\kw{FindDS}}
\newcommand{\dsnode}{\kw{DSnode}}

\SetKwProg{proc}{Procedure}{}{}

\begin{document}

\title{Constant-time Connectivity and 2-Edge Connectivity Querying in Dynamic Graphs%\thanks{Grants or other notes
}


% \author{First Author         \and
%         Second Author %etc.
% }
\author{Lantian Xu\textsuperscript{1}         \and
         Junhua Zhang\textsuperscript{2}%etc.
        \and
        % Jingyi Song %etc.
        % \and
        Dong Wen\textsuperscript{3} %etc.
        \and
         Lu Qin\textsuperscript{1}%etc.
        \and
        Ying Zhang\textsuperscript{1} %etc.
        \and
        Xuemin Lin\textsuperscript{4} 
}


\institute{
   			Dong Wen \Envelope \at
           \email{dong.wen@unsw.edu.au} 
           \and
           Lantian Xu \at
           \email{lantian.xu@student.uts.edu.au}           %  \\
           \and
           Junhua Zhang \at
           \email{zjunhua6@gmail.com} 
           \and
           Lu Qin \at
           \email{lu.qin@uts.edu.au} 
           \and
           Ying Zhang \at
           \email{ying.zhang@uts.edu.au} 
           \and
           Xuemin Lin \at
           \email{xuemin.lin@gmail.com} 
           \and 1 University of Technology Sydney, Australia \\
           2 Northeastern University, China \\
           3 University of New South Wales, Australia \\
           4 Shanghai Jiaotong University, China
}

\date{Received: date / Accepted: date}

\maketitle

\begin{abstract}
  Connectivity query processing is a fundamental problem in graph processing. Given an undirected graph and two query vertices, the problem aims to identify whether they are connected via a path. Given frequent edge updates in real graph applications, in this paper, we study connectivity query processing in fully dynamic graphs, where edges are frequently inserted or deleted. A recent solution, called D-tree, maintains a spanning tree for each connected component and applies several heuristics to reduce the depth of the tree. To improve efficiency, we propose a new spanning-tree-based solution by maintaining a disjoint-set tree simultaneously. By combining the advantages of two trees, we achieve the constant query time complexity and also significantly improve the theoretical running time in both edge insertion and edge deletion. In addition, we extend our connectivity maintenance algorithms to maintain 2-edge connectivity. Our performance studies on real large datasets show considerable improvement of our algorithms.

\end{abstract}

\section{Introduction}
\label{sec:intro}

% Graph is a powerful data structure that is commonly used in computer science and mathematics to capture relationships between objects. It is widely used in applications such as social networks, transportation systems, and circuit designs. 
%
% what is connectivity query
Given an undirected graph, the connectivity query is a fundamental problem and aims to answer whether two vertices are connected via a path. \rr{The connectivity query usually serves as a fundamental operator and is mainly used to prune search space in most
applications. For example, the algorithms to compute paths between two vertices are generally time-consuming.
We can terminate the search immediately if two query vertices are not connected, which avoids unnecessary computation.}
% applications
The problem is driven by a wide range of applications in various fields. For example, the problem helps in identifying whether data packets from a device can reach their destination in a communication network. In biology, testing connectivity between two units in a protein-protein-interaction network can help in understanding the interactions between different components of the system \cite{DBLP:journals/biodatamining/PavlopoulosSMSKASB11}. \rr{For example, existing studies \cite{mao2021digital} trace contacts during the COVID-19 epidemic by modeling an association graph among the high-risk population, the general population, vehicles, public places, and other entities. Two individuals are considered related if there is a path connecting them within a short time window that satisfies certain patterns. Our method can effectively exclude a suspicious vertex pair if they are not connected within the specified time window.} 
% Other applications include business analysis \cite{DBLP:conf/edbt/Dhia12}, web mining \cite{DBLP:journals/vldb/BerendtS00}.

Real-world graphs are highly dynamic where new edges come in and old edges go out. The connectivity between two vertices may change over time. Given the up-to-date snapshot, the connectivity query can be addressed by graph search strategies such as breadth-first search (BFS) and depth-first search (DFS). However, these online search algorithms may scan the whole graph, which is excessively costly for large graphs. 
Instead of online query processing, a straightforward index-based approach is to maintain each connected component as a spanning tree, where two query vertices are connected if they have the same tree root. We may merge two spanning trees for a new edge insertion or split a spanning tree for an old edge deletion.
Index-based methods have been investigated in the literature. \cite{thorup2000near} maintains an Euler tour of a spanning tree, named ET-tree. 
% But its update operations are very expensive. 
\cite{henzinger1995randomized,henzinger1999randomized} improves the ET-tree by terminating early when looking for replacement edges. 
% But its efficiency is still not satisfactory.
% \textcolor{red}{(add some discussion for existing works on connectivity in dynamic graphs?)}
Holm et al. proposed a new method named \hdt\cite{DBLP:conf/stoc/HolmLT98,DBLP:journals/jacm/HolmLT01}, which made the complexity of insertion and deletion $O(\log^{2}{n})$. 
%This was also considered to be the optimal result at the time \cite{marevs2008saga,hanauer2022recent,iyer2001experimental}.
\cite{huang2023fully} further reduced the time complexity slightly to $O(\log{n}(\log{\log{n}})^{2})$.

\begin{table}[t!]
\centering
\begin{tabular}{c c c}
\textbf{Algorithms}                    & \textbf{Ours}          & \textbf{\dtree} \\\hline
Query processing    &   $O(\alpha)$      &   $O(h)$ \\\hline   
Edge insertion      &   $O(h)$      &   $O(h \cdot \kw{nbr_{update}})$\\\hline
Edge deletion       &   $O(h)$      &   $O(h^{2} \cdot \kw{nbr_{scan}})$\\ \hline
\end{tabular}
\vspace{1em}
% \caption{Comparing the average complexity of our method with state-of-the-art. $\alpha$ is a small constant ($\alpha < 5$). $h$ is the average vertex depth in the spanning tree. $d = \frac{m}{n}$ where $m$ is the number of edges and $n$ is the number of vertices. $\kw{nbr_{update}}$ is the time to insert/delete a vertex from the neighbor set. $\kw{nbr_{scan}}$ is the time to scan all neighbors of a vertex.}
\caption{Comparing the average complexity of our method with state-of-the-art. $\alpha$ is a small constant ($\alpha < 5$). $h$ is the average vertex depth in the spanning tree. $\kw{nbr_{update}}$ is the time to insert a vertex into neighbors of a vertex or to delete a vertex from neighbors of a vertex. $\kw{nbr_{scan}}$ is the time to scan all neighbors of a vertex.}
\label{tab:complexity}
\vspace{-1em}
\end{table}

\stitle{The State of the Art.}
%
% \stitle{The State-of-the-Art.}
To improve the connectivity query efficiency practically, \rr{a recent work \cite{DBLP:journals/pvldb/ChenLHB22}, named \dtree, maintains a spanning tree in dynamic graphs, and the spanning tree is represented by maintaining the parent of each vertex. To test the connectivity of two vertices, we search from each vertex to the root and identify whether their roots are the same. The query time depends on the depth of each query vertex, i.e., the distance from the query vertex to the root of the tree. Their main technical contribution is to develop several heuristics to maintain a spanning tree with a relatively small average depth.}
For edge insertion, a non-tree edge is a new edge that connects two vertices in the same tree, and a tree edge is a new edge that connects two vertices in different trees. To insert a non-tree edge $(u,v)$, the \dtree replaces one existing tree edge with the new edge when the depth gap of $u$ and $v$ is over a certain threshold. 
% Additional operations like updating parent pointers and updating subtree size are required to maintain \dtree. 
% The motivation here is to make the tree close to a \rr{BFS tree} which reduces the average depth.
%
To insert a tree edge $(u,v)$ into an index, the \dtree always merges the smaller tree into the bigger tree. Assume that the spanning tree of $u$ is the smaller one. \rr{They rotate the tree (i.e., keep the same tree edges but change the parent-child relationship of certain vertices) so that $u$ becomes the new root. Then, the two trees are connectged by assigning $u$ assigned as the child of $v$.}
\rr{To delete a tree edge $(u,v)$, the spanning tree is immediately split into two trees $T_u$ and $T_v$ for $u$ and $v$, respectively. Assume that the size of $T_u$ is smaller. The \dtree picks $T_u$ and searches non-tree neighbors of each vertex in $T_u$ to find non-tree edges that can reconnect the two trees. The \dtree picks the edge that connects to the shallowest vertex in $T_v$ and reconnects two trees.
Nothing needs to be done when deleting a non-tree edge.
When scanning vertices in both tree updates and query processing of the \dtree, they may rotate the tree if the rotation results in a smaller average depth.}

\stitle{Drawbacks of \dtree.} Much improvement space is left for the \dtree. We summarize the main drawbacks as follows.

\begin{enumerate}
\item \textit{Poor update efficiency.} In \dtree, many efforts are spent on reducing the average depth in tree updates at a high price. Especially in tree-edge deletion, they search non-tree neighbors of all vertices in a sub-tree to identify the replacement edge. The search space can be very large which is not scalable for big graphs.
% This is not scalable for large graphs because the complexity is difficult to bound and the search space can reach half the edge of the graph. 
In addition, the \dtree maintains complex structures to search the replacement edge. Maintaining them incurs extra costs when trees are rotated. 

\item \textit{Poor query efficiency.} In the \dtree, query processing needs to find the roots of two query vertices, and the time complexity of query processing is $h$, where $h$ represents the average depth of vertices. However, in large-scale graphs, the vertex depth in a spanning tree may be very large, and the query efficiency can be very low. 
% A new structure should be proposed to support more efficient queries. However, how to improve existing query indexes is a major challenge.

\end{enumerate}
\noindent

\stitle{Our Approach.} \reftab{complexity} gives a quick view of the theoretical running time of \dtree and our approach. $\kw{nbr_{update}}$ and $\kw{nbr_{scan}}$ depend on the data structure to maintain non-tree neighbors and children. If a balanced binary tree is used, we have $\kw{nbr_{update}} = \log{\avgdeg}$ and $\kw{nbr_{scan}} = \avgdeg$, where $\avgdeg$ is the average degree. If the hash set is used, $\kw{nbr_{update}}$ can be reduced to a constant value, but $\kw{nbr_{scan}}$ will be $\avgdeg+b$ where $b$ is the number of buckets for a hash set. 
By comparison, our approach achieves the almost-constant query time and reduces the time complexity of both insertion and deletion simultaneously. 
% Note that for edge deletion, it is clear to see that $O(\frac{1}{d-1}) < O(\kw{nbr_{scan}})$ given $d > 1$ and $\kw{nbr_{scan}} \ge d$ for real graphs.

We start by considering how to improve the query efficiency in theory.
Our idea is inspired by the disjoint-set data structure which organizes items in different sets.
% Disjoint-set is a well-known data structure which is friendly to insertion-only updates. 
It offers two operators: \kw{Find} identifies the representative of a set containing the element, and \kw{Union} merges two sets. The structure is commonly identified to achieve the amortized constant time complexity for \kw{Find} and \kw{Union}, which correspond to querying connectivity and edge insertion in our case, respectively. However, it is hard to handle deletion situations if we only use the disjoint set. When an edge is deleted only based on the disjoint set, it is not able to identify if the connected component is disconnected, and we need to compute the connected component from scratch as a result.
% On the one hand, it is difficult to judge whether deleting edges will cause splitting of connected components because disjoint-set does not maintain non-tree edges. On the other hand, when connected components need to be split, it is difficult to split a set into two sets.
To achieve the constant query time complexity while keeping the high efficiency for both edge insertion and deletion, we maintain a spanning tree and a disjoint set simultaneously and combine the advantages of two data structures. The spanning tree implementation in our algorithm is called \basetree, and the disjoint set implementation in our algorithm is called \dstree.

% ========
% \begin{enumerate}
% \item \textcolor{red}{\textit{Querying in \dstree.} We use \dstree for query processing to achieve the optimal query time complexity. \dstree also helps speed up edge insertion for \basetree.}

% \item \textit{Edge Deletion.} For edge deletion that disconnects a connected component, we use \basetree to identify the result connected components which facilitate splitting \dstree into two trees for two resulting connected components, respectively.

% \item \textit{Updating \dstree with the help of \basetree.} For edge deletion that disconnects a connected component, we use \basetree to identify the result connected components which facilitate splitting \dstree into two trees for two resulting connected components, respectively.
% \end{enumerate}
% ========

\stitle{Spanning Tree Implementation.}
Our \basetree is extended from \dtree by applying several modifications for higher practical efficiency. 
% Our spanning tree structure is called \basetree.
% The spanning tree maintained in our algorithms is an improved version of \dtree, called \basetree, to improve the update processing. 
%We observe that the vertex depth is still well-bounded theoretically even if we do not apply any heuristic to reduce it. Motivated by this, we avoid certain expensive operations for reducing depth in \dtree. 
Our results show that our improved version is much more efficient than \dtree in all aspects. We make the following improvements in response to the drawbacks of \dtree. 
We apply a new heuristic early-termination technique to derive a better theoretical bound when searching replacement edges for a deleted tree edge. Given a tree $T$, after deleting a tree edge, a subtree $T_{1}$ with a larger size and a subtree $T_{2}$ with a smaller size are formed. We find the replacement edge with lowest depth in the small subtree $T_{2}$. This allows us to terminate early and improve the efficiency of searching the replacement edge.
% and \textcolor{blue}{we prove that the searching time is never related to $m$, where $m$ is the number of edges in the graph and \dtree may takes $O(m)$ to search a replacement edge.} 
In addition, we also avoid maintaining non-tree neighbors and children for each vertex which reduce the maintenance overhead but not sacrifice the update efficiency.

\stitle{Disjoint Set Implementation.}
To support edge deletion, we implement the children of each tree vertex in the disjoint set as a doubly linked list. We design a set of operators to delete an item from the disjoint set in constant time while keeping the constant time of \kw{Find} and \kw{Union}. When a connected component $A$ is disconnected into $B$ and $C$ after deleting an edge, we use our \basetree to identify them (assuming $|B| \le |C|$). In the disjoint set, we remove all vertices of $B$ from $A$ in $O(|B|)$ time and union all vertices of $B$ to create a new connected component. In this way, the time complexity of updating our disjoint set is bounded by that of the spanning tree.

\stitle{Maintaining 2-edge Connectivity.} We analyze the relationship between connected components and 2-edge-connected components and propose a method for maintaining 2-edge connectivity on a spanning tree by tracking the number of replacement edges for each tree edge. By further combining this with the disjoint-set structure, we achieve constant-time queries for 2-edge connectivity.

\stitle{Contributions.} We summarize our main contributions as follows.

\begin{enumerate}
\item \textit{Theoretical almost-constant query efficiency.} We propose a new approach to combine the advantages of spanning tree and disjoint-set tree. The approach achieves amortized constant query time without sacrificing the time complexity of both edge insertion and edge deletion. Our final solution is theoretically more efficient than the state-of-the-art \dtree in all aspects. By combining the disjoint-set tree, our 2-edge connectivity maintenance algorithm also achieves constant-time queries. 

\item \textit{Theoretical higher update efficiency.} We propose a new spanning-tree structure for higher updating efficiency in both theory and practice compared with \dtree. We also propose a new disjoint-set data structure to handle the edge deletion together with the spanning-tree. Our algorithms are highly efficient for both connectivity updates and 2-edge connectivity updates.
% propose a spanning-tree based solution to answer connectivity queries in dynamic graphs. Compared with the state of the art, our method is more efficient in all aspects, i.e., query processing, edge insertion, and edge deletion.

\item \textit{Outstanding practical performance.} We conduct extensive experiments on many real datasets in various settings. The results demonstrate the higher practical efficiency of our algorithms compared with the state of the art.
\end{enumerate}

% \stitle{Organization.} The rest of the paper is organized as follows. We first define the research problem in \refsec{pre}. \refsec{exist} introduce the state-of-the-art solution \dtree. \refsec{base} proposes the improved version of \dtree. \refsec{opt} proposes our final solution for dynamic connectivity queries. \refsec{exp} evaluates our approach with extensive experiments. \refsec{rel} reviews other related works, and we conclude the paper in \refsec{conclu}.
% %
% \noindent 
% Due to space limitation, we omit some examples, proofs, and experiments which can be found in our technical report \cite{report}. 

%!TEX root = main.tex

\vspace{-1em}
\section{Preliminary}
\label{sec:pre}

%introduce graph notations
Given an undirected simple graph $G(V,E)$, $V$ and $E$ denote the set of vertices and the set of edges, respectively. We use $n$ and $m$ to denote the number of vertices and the number of edges, respectively, i.e., $n=|V|, m=|E|$. The neighbors of a vertex $u$ is represented by $\nbr(u)$, i.e., $\nbr(u)=\{v \in V\mid(u,v) \in E\}$. 
A tree $T$ is a connected graph without any cycle. \rr{
%We use the term "node" to refer to the vertex in a tree for clearance given that a tree node in this paper usually maintains additional attributes compared with a single vertex ID.
% $V(T)$ and $E(T)$ represent tree vertices and tree edges, respectively. 
Given a tree $T$, $\parent(u)$ denotes the parent vertex of a vertex $u$ in the tree. $\depth(u)$ denotes the depth of the vertex $u$, i.e., the number of tree edges from $u$ to the tree root. $\height$ denotes the average depth of all vertices in the tree, i.e., $\height(T) = \sum_{u \in T}\depth(u)/|T|$. Each tree has only one root. Rotating a tree means keeping the same set of tree edges but changes the tree root. As a result, the parent of each vertex from the old root to the new root becomes the child of the vertex.}
% "Rotate" denotes changing the parent-child relationship between vertex without changing the edges in the tree. "Reroot" denotes making the root change by rotating the tree. $\depth_T(u)$ denotes the depth of $u$ in $T$, and $\height$ denotes the average depth, i.e., $\height(T) = \sum_{u \in T}\depth_T(u)/|T|$. 
% We omit the subscript $T$ when it is clear from the context. 
% \rr{We use the term "node" to refer to the data structure (include necessary attributes such as parent-child relationship) of a vertex in a tree.}
%
% We summarize frequently used notations in \reftab{notation}. 
%
A path $P = \langle v_1,v_2,...,v_l \rangle$ in $G$ is a sequence of vertices in which each pair of adjacent vertices are connected via an edge, i.e., for all $1 \le i < l$, $(v_i,v_{i+1}) \in E$. We say two vertices $u,v$ are connected if there exists a path such that $u$ and $v$ are terminals. A connected component (CC for short) in a graph is a maximal subgraph in which every pair of vertices are connected. Therefore, two vertices are connected if they are in the same connected component.

%\vspace{-0.5em}
\begin{definition}
\textsc{(Connectivity Query)} Given a graph $G$ and two query vertices, the connectivity query aims to determine whether $u$ and $v$ are connected in $G$.
\end{definition}
%\vspace{-0.5em}

% Given that real graphs are often highly dynamic, we study the connectivity query in dynamic graphs where edges are frequently inserted or deleted.

% \vspace{0.5em}
\stitle{Problem Definition.} Given a graph $G$, we aim to develop an index for processing connectivity queries between arbitrary pairs of vertices and maintain the index when a new edge is inserted or an existing edge is deleted.

%!TEX root = main.tex
\section{Existing Solutions}
\label{sec:exist}

\subsection{Query Processing in Static Graphs}
A straightforward online method for the connectivity query is to perform a bidirectional breath-first search (BFS) or a depth-first-search from a query vertex. Once meeting the other query vertex in the search, we identify that two vertices are connected. The online method for one connectivity query takes $O(m)$ time in the worst case, which is hard to be tolerated in large graphs. To improve the query efficiency, we can index an identifier of the corresponding connected component for every vertex in a static graph, which takes $O(n)$ space. In this way, the connectivity query can be answered by checking the identifier of two query vertices which takes constant time complexity.
Compared with static graphs, dealing with fully dynamic graphs for efficient connectivity query processing is much more challenging. For instance, removing an edge may disconnect the connected component. Certain techniques are expected to identify the connectivity of the original connected component and disconnect the component index structure accordingly.

\begin{figure}[t!]
\centering
\includegraphics[width=0.2\columnwidth]{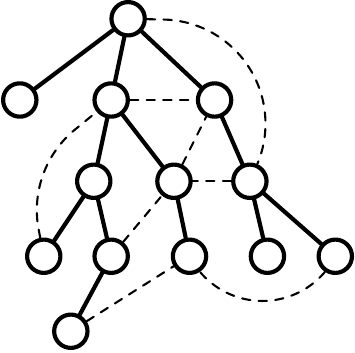}
\vspace{-1em}
\caption{An example graph and a possible spanning tree.}
\label{fig:spanning}
\end{figure}

\vspace{-0.5em}
\subsection{Maintaining Spanning Trees}

A straightforward index-based solution for connectivity queries in dynamic graphs is to maintain a spanning tree for each connected component.
% explain spanning tree
Given a connected component $C$, a spanning tree is a connected subgraph of $C$ including all vertices with the minimum number of edges. The subgraph is a tree structure clearly. An example graph and its spanning tree are presented in \reffig{spanning}. All tree edges are represented as solid lines, and all non-tree edges are represented as dashed lines. 
%
%
% query processing
Given two query vertices, we can locate the root of each vertex by continuously scanning the tree parent. If the tree roots of two query vertices are the same, they are in the same spanning tree and the same connected component.

Maintaining simple spanning trees for all connected components is not challenging. For inserting a new edge $(u,v)$, nothing needs to be done if two vertices have the same root. We call this case \textit{non-tree edge insertion}. Otherwise, they are from different spanning trees, and we need to merge them. We call this case \textit{tree edge insertion}. To merge two trees, we pick one vertex $u$ of the inserted edge and rotate its spanning tree $T_u$ to make $u$ be the root. Then we can add $T_u$ as a child subtree of $v$ in its spanning tree $T_v$.
For deleting an edge $(u,v)$, nothing happens if $(u,v)$ is a non-tree edge. Deleting a tree edge will immediately divide the spanning tree into two smaller trees. Then we need to identify if there is another edge connecting the two spanning trees. Once such an edge is found, we process it as the tree edge insertion.

\vspace{-0.5em}
\subsection{The State-of-the-Art}
\label{subsec:dtree}

Recently, Qing et al. \cite{DBLP:journals/pvldb/ChenLHB22} proposed a solution for connectivity query processing in dynamic graphs. Their solution is called \dtree. They maintain a spanning tree with additional properties for each connected component to improve the average query efficiency.
Their method is based on the following lemmas.
% \stitle{Reducing Average Tree Height.} To identify the efficiency of query processing, they derive the following lemmas.

\vspace{-0.5em}
\begin{lemma}
The average costs of evaluating connectivity queries by spanning trees is optimal if each tree $T$ minimizes $S_d(T)$, where $S_d(T)$ is the sum of distances between root and descendants in $T$, 
i.e., $S_d(T)=\sum_{u \in V(T)}\depth(u)$. \cite{DBLP:journals/pvldb/ChenLHB22}
\end{lemma}
\vspace{-0.5em}

% \vspace{-0.5em}
\begin{definition}
\textsc{(Centroid)} Given a spanning tree $T$, a centroid of $T$ is a vertex with the smallest average distance to all other vertices.
\end{definition}
\vspace{-0.5em}

\vspace{-0.2em}
\begin{lemma}
\label{lem:dtree}
The average cost of each connectivity query by spanning trees is optimal if each tree is 1) rooted in the centroid, and 2) a BFS tree, i.e., the distance from every vertex to the root is minimal. \cite{DBLP:journals/pvldb/ChenLHB22}
\end{lemma}
\vspace{-0.5em}

It is already very expensive to maintain just a valid BFS tree or a spanning tree rooted in the centroid. Therefore, \dtree develops several strategies to reduce the average depth in tree maintenance. We summarize them into two categories.

\stitle{Centroid Heuristic.} The centroid heuristic rotates spanning trees for certain cases and aims to reduce the average depth without changing the tree edges. Specifically, when scanning vertices in tree updates, they attempt to locate the centroid of the spanning tree and rotate the tree so that the centroid is the root of the tree. To this end, they observe that given a spanning tree $T$ rooted in its centroid $c$, the subtree size of every child of $c$ in $T$ is not larger than $|T|/2$. Therefore, once the subtree size of vertex $u$ in $T$ is larger than $|T|/2$ and all children of $u$ are not, they identify that $u$ is close to the centroid and make the tree rooted in the new root $u$. 

\stitle{BFS Heuristic.} BFS heuristics reduce the average depth when updating the tree structure (i.e., changing tree edges), which happens in both edge insertion and edge deletion. When it is required to merge one tree $T$ into the other tree $T'$, the BFS heuristic pick a vertex $u$ in $T$ that has a non-tree neighbor $v$ in $T'$ with the lowest depth in $T'$. Then it rotates $T$ to be rooted in $u$ and adds the tree $T$ as a subtree of $v$.

\stitle{Algorithms of \dtree.} We describe the process of \dtree below. In addition to maintaining the parent and the corresponding vertex id for each tree vertex like the spanning tree, \dtree maintains the subtree size, children, and non-tree neighbors for each vertex. Subtree size is used for the centroid heuristic. Children and non-tree neighbors are used to efficiently find a replacement edge when deleting a tree edge. For query processing, as a by-product of continuously scanning parents for each vertex, they check if the visited child of the root has over half tree vertices in the subtree. If so, they apply the centroid heuristic and rotate the tree. Their query time complexity is $O(h)$ where $h$ is the average depth.

For non-tree edge insertion, they check if the depth gap between two vertices $u,v$ of the edge is larger than $1$. If so, they pick one tree edge between $u$ and $v$ to cut the tree and apply the BFS heuristic to merge two trees by connecting $u$ and $v$. For tree edge insertion, they pick and rotate a smaller tree and merge it into the bigger one. Note that every time a tree is updated, we need to update the subtree size for influenced vertices. Meanwhile, we rotate the tree once we meet a vertex satisfying the centroid property. They also need to maintain children and non-tree neighbors for each vertex even though they are not used in edge insertion. The time complexity of edge insertion time is $O(h \cdot \kw{nbr_{update}})$, where $\kw{nbr_{update}}$ is the time complexity to insert or delete an item from the children set and non-tree neighbors. If a hash set structure is used, $\kw{nbr_{update}}$ can be reduced to a small constant but much memory will be used for a large number of buckets. If a balanced binary search tree is used, the running time is $O(h\cdot\log{d})$ where $d$ is the average degree.

The tree is immediately split into two trees when deleting a tree edge. They search non-tree neighbors of all vertices in the smaller tree to find a replacement edge to connect two trees. Given multiple replacement edge candidates, they apply the BFS heuristic to link the tree to the vertex with the lowest depth. The time complexity of edge deletion is $O(h \cdot \kw{nbr_{scan}})$ if non-tree neighbors of each vertex can be scanned in linear time. Note that $h$ is the average depth and is also the average subtree size of each vertex in the tree.
The dominating cost is to scan non-tree neighbors for all vertices in the smaller tree.
There are still some Euler tour-based existing solutions, such as \hdt \cite{DBLP:conf/stoc/HolmLT98,DBLP:journals/jacm/HolmLT01}. However, as shown in our experimental results, \dtree is more efficient than \hdt. Therefore, we mainly introduce \dtree here. We will also introduce some other related works in \refsec{rel}.

%!TEX root = main.tex
\section{Revisiting D-Tree}
\label{sec:base}

\dtree makes many efforts to reduce the average depth in the spanning tree. However, certain heuristics yield marginal benefits for the average depth but sacrifice much efficiency as a trade-off. To improve the overall efficiency, we follow the framework of \dtree and propose an improved lightweight version called \baseidxname (\underline{I}mproved \underline{D}ynamic Tree) in this section. We discuss our implementations for edge insertion and edge deletion in this section. 
% We will discuss a constant-time querying algorithm in \refsec{opt}.

% ideas in \refsubsec{motivation}. Details for edge insertion and edge deletion will be introduced in \refsubsec{baseinsert} and \refsubsec{basedelete}, respectively. We will discuss a constant-time query algorithm in \refsec{opt}.

\vspace{-0.5em}
\subsection{Motivation}
\label{subsec:motivation}

Our solution is motivated by the following observations.

\vspace{-0.5em}
\begin{observation}
Centroid heuristics in query processing of \dtree may not help with improving query efficiency.
\end{observation}
\vspace{-0.5em}

As introduced in \refsubsec{dtree}, \dtree may rotate the tree in query processing to make the new root close to the centroid. Even though it just additionally brings constant time complexity, avoiding unnecessary tree updates may improve the query efficiency. One may suspect that avoiding this step will increase the average depth, which reduces the average query efficiency to some extent. We respond with the following lemma.

\vspace{-0.5em}
\begin{lemma}
\label{lem:centroid}
Given a spanning tree $T$ rooted in $u$, let $c$ be the centroid in $T$ and $T'$ be the tree with $c$ as the root and the same tree edges as $T$. We have $maxdepth_T \le 2\cdot maxdepth_{T'}$ where $maxdepth_T$ represents the largest vertex depth in $T$.
\end{lemma}
\vspace{-0.5em}

\begin{proof}

    \rr{For any two vertices $u$ and $v$ in tree $T'$ with the centroid $c$ as the root, their distances to the centroid $c$ must be less than $\depth_{T'}(v)$. Therefore, there must exist a path passing through $c$ between $u$ and $v$, and the length of this path is less than $2\cdot\depth_{T'}(v)$. Therefore, the distance between any two vertices in $\depth_T(v)$ must be less than $2\cdot\depth_{T'}(v)$.}
\end{proof}

\reflem{centroid} shows that the theoretical max depth is still well-bounded even if we never rotate and balance the tree. On the other hand, \dtree cannot guarantee to always maintain the centroid as the root. It is a trade-off between the depth and the additional efforts to reduce it.
Note that the depth determines not only the query efficiency. Tree updates can also benefit from low depth since scanning from a vertex to the root happens in both edge insertion and deletion. Therefore, in our implementation, we still perform some rotation operations to make the root close to the centroid but never update the tree in query processing.

\stitle{Improved Query Processing.} Our query algorithm of \baseidxname is called \algbasequery. The algorithm only searches the tree root of each query vertex and identifies if their roots are the same.
Our performance studies show that the average depth of \baseidxname is competitive to that of \dtree, and our improved query algorithm \algbasequery is more efficient in most real datasets. We will further eliminate the dependency between the query efficiency and the average depth in \refsec{opt}, which achieves almost-constant query efficiency.

\vspace{-0.5em}
\begin{observation}
Search replacement edges in processing tree edge deletion of \dtree is expensive.
\end{observation}

\vspace{-0.5em}

Edge deletion in \dtree takes much more time compared with edge insertion. \rr{It is caused by searching the replacement edge to reconnect trees when deleting a tree edge. 
%They search non-tree neighbors for the whole subtree to find a connected vertex with the lowest depth. 
Assume a tree edge $(u,v)$ is deleted where $v$ is the parent of $u$. \dtree searches the neighbors of all vertices in the subtree rooted by $u$ to find all edges that can reconnect the tree. Then they pick one of them based on certain heuristics. As a result, almost half edges in the graph will be scanned in the worst case. 
% The requirement of searching replacement edges makes edge deletion taking much more time than edge insertion in \dtree. 
Our main target is to significantly improve the efficiency of edge deletion without sacrificing the efficiency of query processing and edge insertion. Our method can terminate searching the subtree once any replacement edge is found.}

\vspace{-0.5em}
\begin{observation}
Maintaining children and non-tree neighbors for each vertex is expensive.
\end{observation}
\vspace{-0.5em}

\dtree maintains children and non-tree neighbors for finding the replacement edge as mentioned above. In exchange, the two sets require updating every time we rotate the tree or replace tree edges which frequently happens in query processing, edge insertion, and edge deletion. Maintaining them for each vertex is costly.
Based on these observations, our implementation \baseidxname only maintains a subset of attributes in \dtree for each vertex including: 

\begin{itemize}
\item $\parent(u)$: the parent of $u$ in the tree;
\item $\stsize(u)$: the size of subtree rooted in $u$.
\end{itemize}

\begin{algorithm}[t!]
\caption{\algbaseinsert}
\label{alg:baseinsert}
% \SetAlgoVlined
\KwIn{a new edge $(u,v)$ and the \baseidxname index}
\KwOut{the updated \baseidxname}

$root_u \leftarrow$ compute the root of $u$\;
$root_v \leftarrow$ compute the root of $v$\;

\tcc{non-tree edge insertion}
\If{$root_u = root_v$}{
    \lIf{$\depth(u)<\depth(v)$}{$\kw{swap}(u,v)$}
    \lIf{$\depth(u)-\depth(v) \le 1$}{\Return}
    \tcc{reduce tree deviation}
    
    $w \leftarrow u$\;
    \For{$1 \le i < \frac{\depth(u)-\depth(v)}{2}$}{
        $w \leftarrow \parent(w)$\;
    }
    $\algunlink(w)$\;
    $\alglink(\algreroot(u),v,root_v)$\;
    \Return\;

}

\tcc{tree edge insertion}
\If{$\stsize(root_u) > \stsize(root_v)$}{
    $\kw{swap}(u,v)$\;
    $\kw{swap}(root_u,root_v)$\;
}
$\alglink(\algreroot(u),v,root_v)$\;

\end{algorithm}

\vspace{-0.5em}
\subsection{Edge Insertion}
\label{subsec:baseinsert}

Our improved algorithm for edge insertion is called \algbaseinsert. Before presenting details, we introduce three tree operators. They are the same as those in \dtree except excluding updates for children and non-tree neighbors.
%Some other important operators are shown in our technical report \cite{report}.

\begin{itemize}
\item \rr{$\algreroot(u)$ rotates the tree and makes $u$ as the new root. It updates the parent-child relationship and the subtree size attribute from $u$ to the original root. The time complexity of $\algreroot()$ is $O(\depth(u))$.}

\item \rr{$\alglink(u,v,root_v)$ adds a tree $T_u$ rooted in $u$ to the children of $v$. $root_v$ is the root of $v$. Given that the subtree size of $v$ is changed, it updates the subtree size for each vertex from $v$ to the root. We apply the centroid heuristic by recording the first vertex with a subtree size larger than $\stsize(root_v)/2$. If such a vertex is found, we reroot the tree, and the operator returns the new root. The time complexity of $\alglink()$ is $O(\depth(v))$.}

\item \rr{$\algunlink(u)$ disconnect the subtree of $u$ from the original tree. All ancestors of $u$ are scanned to update the subtree size. The time complexity of $\algunlink()$ is $O(\depth(u))$.}
\end{itemize}

\noindent We present the pseudocode of the improved edge insertion in \refalg{baseinsert}. Most processes are the same as that of \dtree. We first conduct a connectivity query to identify if two vertices are in the same tree. Lines 3--11 apply the BFS heuristic for non-tree edge insertion. When the depth gap between $u$ and $v$ is over $1$, a major difference from \dtree is about the strategy of BFS heuristic in Line 7. \dtree uses the threshold $\depth(u)-\depth(v)-2$ instead of $\frac{\depth(u)-\depth(v)}{2}$. To reduce the average depth, we add half vertices from $u$ to its ancestor with the same depth of $v$ to the subtree rooted in $u$ (Lines 7--9). Then we add the updated tree rooted in $u$ to the children of $v$. 

\begin{figure}[t!]
\centering
\vspace{-0.5em}
\includegraphics[width=0.5\columnwidth]{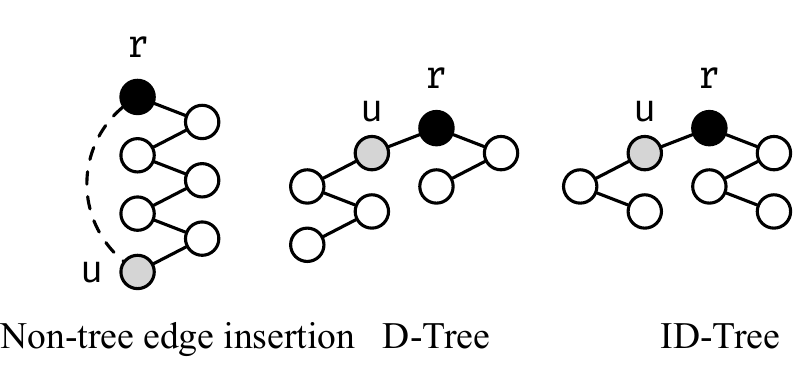}
\vspace{-1em}
\caption{Non-tree edge insertion in \dtree and \basetree.}
\label{fig:intra}
\end{figure}
% \begin{figure}[t!]
% \centering
% \includegraphics[width=0.8\columnwidth]{inter.pdf}
% \vspace{-0.8em}
% \caption{An example of tree edge insertion for \basetree.}
% \label{fig:inter}
% \end{figure}

\vspace{-0.5em}
\begin{example}
\reffig{intra} shows the different strategies of BFS heuristic between \dtree and our \basetree. We insert a non-tree edge $(u,r)$. We have $\depth(u) = 6$ and $\depth(r) = 0$. Based on the strategy of \dtree, we add $u$ together with its three ancestors as a child subtree of $r$, which is presented in the middle figure. However, based on our strategy, we add $u$ together with its two ancestors as a child subtree of $r$, which is presented in the right figure. The average depth of the tree is reduced from $1.9$ to $1.7$ in this example.
\end{example}
\vspace{-0.5em}

For tree edge insertion (Lines 12--14), our process is the same as \dtree, which merges a smaller tree into a bigger tree. 

\begin{figure}[t!]
\centering
\includegraphics[width=0.5\columnwidth]{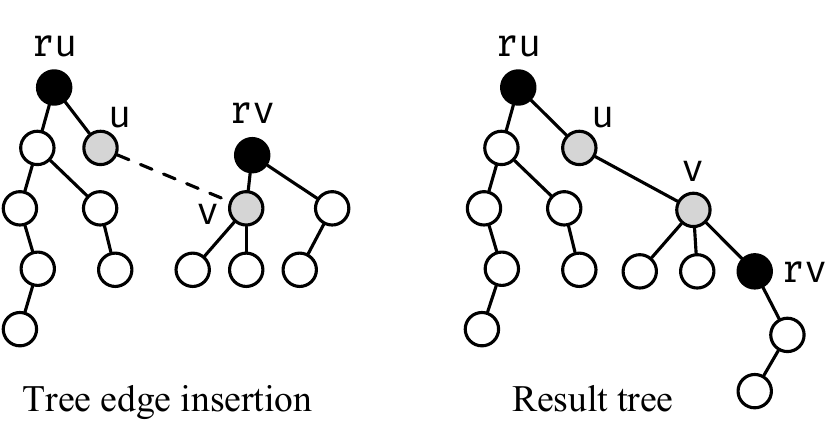}
\vspace{-0.8em}
\caption{\rr{An example of tree edge insertion for \basetree.}}
\label{fig:inter}
\end{figure}
\vspace{-0.2em}
\begin{figure}[t!]
\vspace{-0.3em}
\centering
\includegraphics[width=0.5\columnwidth]{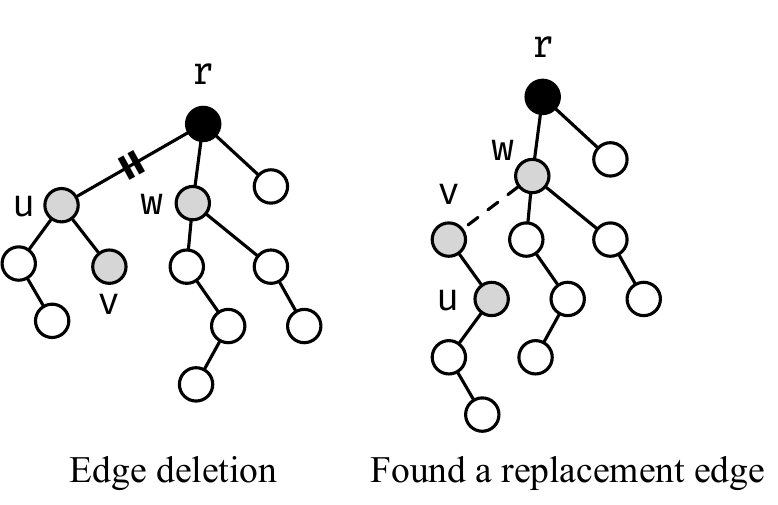}
\vspace{-1em}
\caption{An example of tree edge deletion for \basetree.}
\label{fig:delete}
\end{figure}

\vspace{-0.5em}
\begin{example}
\rr{An example for tree edge insertion is shown in \reffig{inter}. After inserting $(u,v)$, we reroot the tree rooted by $rv$ to $v$ and add the updated tree rooted by $v$ to the children of $u$.}
\end{example}
\vspace{-0.5em}

% \vspace{-0.5em}
% \begin{example}
% An example for tree edge insertion is shown in \reffig{inter}. After inserting $(u,v)$, we reroot the tree rooted in $rv$ to $v$ and add the updated tree rooted in $v$ to the children of $u$.
% \end{example}
% \vspace{-0.5em}

% \begin{theorem}
% \label{thm:baseinsert}
% The time complexity of \refalg{baseinsert} is $O(h)$, where $h$ is the average depth.
% \end{theorem}

% \begin{proof}
%     First, computing the roots of $u$ and $v$ requires $O(h)$ time. Then, for non-tree edge insertion, if the depth gap between $u$ and $v$ is less than 1, no other operation is needed; otherwise, an existing edge needs to be unlinked and a new inserted edge needs to be linked, which also takes $O(h)$ time. For tree edge insertion, the two trees can be linked directly after rerooting, which also takes $O(h)$ time. Therefore, the overall time complexity of \refalg{baseinsert} is still $O(h)$.
% \end{proof}

\vspace{-0.7em}
\subsection{Edge Deletion}
\label{subsec:basedelete}

\begin{algorithm}[t!]
\caption{\algbasedelete}
\label{alg:basedelete}
\KwIn{an existing edge $(u,v)$ and the \baseidxname}
\KwOut{the updated \baseidxname}

\lIf{$\parent(u) \neq v \land \parent(v) \neq u $}{\Return}
\lIf{$\parent(v) = u$}{$\kw{swap}(u,v)$}
$root_v \leftarrow \algunlink(u)$\;

\tcc{reduce the worst-case time complexity of searching replacement edge in subtree}
\lIf{$\stsize(root_v) < \stsize(u)$}{$\kw{swap}(u,root_v)$}
    % $root_v \leftarrow v$\;

\tcc{search subtree rooted in $u$}
$Q \leftarrow$ an empty queue, $Q.push(u)$\;
$S \leftarrow \{u\}$\;
\tcc{$S$ maintains all visited vertices}

\While{$Q \neq \emptyset$}{
    $x \leftarrow Q.pop()$\;
    \ForEach{$y \in \nbr(x)$}{
        \lIf{$y = \parent(x)$}{\textbf{continue}}
        \uElseIf{$x = \parent(y)$}{
            $Q.push(y)$\;
            $S \leftarrow S \cup \{y\}$\;

        }
        \Else{
            $succ \leftarrow \kw{true}$\;
            \ForEach{$w$ from $y$ to the root}{

                \uIf{$w \in S$}{
                    $succ \leftarrow \kw{false}$\;
                    $\textbf{break}$\;
                }
                \Else{
                    $S \leftarrow S \cup \{w\}$\;
                    % mark $w$ as visited\;
                    % $succ \leftarrow \kw{true}$\;
                }
                
            }
            \If{$succ$}{
                $root_v \leftarrow \alglink(\algreroot(x),y,root_v)$\;
                \Return\;
            }

            % $\algreroot(x)$\;
            % $\parent(x) \leftarrow y$\;
            % \ForEach{$w$ from $y$ to the root}{
            %     $\stsize(y) \leftarrow \stsize(y)+\stsize(x)$\;
            % }
        }
    }
}
\end{algorithm}

Our improved edge deletion algorithm is called \algbasedelete. The pseudocode is presented in \refalg{basedelete}. Similar to \dtree, we need to search for a replacement edge when a tree edge is deleted, and we always search from the smaller tree (Line 4).
Compared with \dtree, our strategy to find a replacement edge is completely different. The general idea is to immediately terminate the search once a replacement edge is found. We use a queue $Q$ to maintain all visited vertices. In Line 6, $S$ maintains the set of all visited vertices in the subtree of $u$. Without children and non-tree neighbors, we directly search graph neighbors (Line 9) for each vertex popped from the queue. Lines 11--13 are for children. We add them to the queue for further exploration. For each non-tree neighbor $y$ (Lines 14--24), we check if $y$ is from the different tree by scanning the ancestors of $y$ (Lines 16--21). An optimization here is to terminate the current iteration once we reach a vertex recorded in the subtree of $u$ (Lines 17--19).
If $succ$ keeps positive after all iterations of Lines 16--21, it means we already reach the root of $y$, and $y$ belongs to a different tree. As a result, we link two trees via the edge $(x,y)$ (Lines 22--24). Otherwise, all visited vertices in Line 17 are in the subtree of $u$, and we continue to search for the next possible replacement edge. Maintaining all visited vertices in $S$ guarantees each vertex is scanned only once for the whole process of searching for replacement edges.

\vspace{-0.5em}
\begin{example}
A running example of tree edge deletion is shown in \reffig{delete}. In the left figure, $(r,u)$ is deleted, and the tree is split into two subtrees. We search for the replacement edge from the smaller tree, which is rooted in $u$. Assume that a replacement edge $(v,w)$ is found. We reroot the smaller tree to $v$ and link two trees by connecting $(v,w)$. The right figure shows the final result.
\end{example}
\vspace{-0.5em}

\section{Theoretical Analysis}
\label{sec:opt}

\vspace{0.9em}
\subsection{Edge Insertion}

\vspace{0.7em}
\begin{theorem}
\label{thm:baseinsert}
The time complexity of \refalg{baseinsert} is $O(h)$, where $h$ is the average depth.
\end{theorem}
\vspace{-0.5em}

\vspace{-0.5em}
\begin{proof}
    Computing the roots of $u$ and $v$ requires $O(h)$ time. For non-tree edge insertion, if the depth gap between $u$ and $v$ is less than 1, no other operation is needed; otherwise, an existing edge needs to be unlinked and a new inserted edge needs to be linked, which also takes $O(h)$ time. For tree edge insertion, the two trees can be linked directly after rerooting, which also takes $O(h)$ time. 
    % Therefore, the overall time complexity of \refalg{baseinsert} is still $O(h)$.
\end{proof}

\subsection{Edge Deletion}

Our strategy to find a replacement edge is simple but efficient in both theory and practice. Even though we scan graph neighbors for a vertex from the queue, we show that the search space is quite small by the following lemma. We assume that the distribution of vertices and edges is uniform.

\vspace{-0.5em}
\begin{lemma}
\label{lem:expect}
The expected total number of iterations of Line 9 in \refalg{basedelete} is $O(h)$, where $h$ is the average tree depth. 
%\rr{without referring Dtree, Really?}
\end{lemma}
\vspace{-0.5em}

\vspace{-0.5em}
\begin{proof}
    % Given a graph $G$, with $O(n)$ tree edges in its spanning forest, there are $O(m-n)$ non-tree edges. In line 9 of \refalg{basedelete}, we need to traverse each neighbor $y$ of $x$. The probability that $(x,y)$ is a non-tree edge is $\frac{m-n}{m}$. Since the BFS search is always performed on the tree with the smaller $st\_size$ rooted at $u$, each non-tree edge has a probability of being the replacement edge greater than 1/2.
    % \begin{eqnarray}    \label{eq}
    % E&=&\frac{2m}{m-n}=\frac{2d}{d-1}  \nonumber    \\
    % ~&=&2+\frac{2}{d-1}=O(\frac{1}{d-1})
    % \end{eqnarray}
    \rr{The sum of the subtree size of all vertices in a tree is equal to the sum of the depth of all vertices. Therefore, the average subtree size of a vertex can be considered as $O(h)$. In \refalg{basedelete}, we need to traverse tree edges and non-tree edges in the subtree. The tree edges are visited at most $O(h)$ times.} For non-tree edges, since the BFS search is always performed on the tree with the smaller $st\_size$ rooted at $u$, each non-tree edge has a probability of being the replacement edge greater than 1/2. Therefore, the expected total number of iterations of Line 9 is still $O(h)$.
\end{proof}
\vspace{-0.5em}

Based on \reflem{expect}, we show the time complexity for our improved algorithm for edge deletion.

\vspace{-0.5em}
\begin{theorem}
\label{thm:idelete}
\rr{The expected time complexity of \refalg{basedelete} is $O(h)$.}
\end{theorem}
\vspace{-0.5em}

\vspace{-0.5em}
\begin{proof}
    \rr{For each iteration of Line 16 in \refalg{basedelete}, we traverse the path from $y$ to the root, which takes $O(h)$ time. Based on \reflem{expect}, the expected number of visits to non-tree edges is bounded by $2$. Therefore, the expected time complexity of \refalg{basedelete} is $O(h)$.}
\end{proof}

% For real world dense graphs, $d$ can be very large, and $\frac{1}{d-1}$ is a constant value. The time complexity of \refalg{basedelete} can be represented as $O(h)$, which achieves a considerable speedup compared with $O(h \cdot \kw{nbr_{scan}})$ in \dtree.

%!TEX root = main.tex

\section{Constant-time Query Processing}
\label{sec:opt}

The ideas of \dtree and our improved version in \refsec{base} are to maintain a near balanced spanning tree since the query time depends on the depth of each vertex. In this section, we eliminate the dependency and \textbf{further improve the query efficiency to almost-constant}. Meanwhile, we \textbf{bound the same theoretical running time for edge insertion and edge deletion}.

\vspace{-0.7em}
\subsection{Utilizing the Disjoint-Set Structure}
\vspace{-0.2em}

Our idea is inspired by the disjoint-set data structure \cite{DBLP:journals/jacm/TarjanL84}. It provides two operators, $\kw{Find}$ and $\kw{Union}$, which are presented in \refalg{ds}. $\kw{Find}$ returns an id for the set containing the input item, and $\kw{Union}$ merges the sets of two given items. The structure maintains each set as a tree and uses the tree root to represent the set. Two crucial optimizations are adopted. The first is path compression. 
%In the path from a query item to the root in \kw{Find}, the algorithm connects every visited item to the root directly in the path compression optimization.
When path compression is applied, in the path from a query item to the root in \kw{Find}, the algorithm connects every visited item to the root directly.
The second operation is union by size. Given that the $size$ of each item in \kw{Union} represents the number of items in the tree, we always merge the smaller tree with the larger tree. With these two optimizations, the amortized time complexity for both $\kw{Find}$ and $\kw{Union}$ is $O(\alpha(n))$, where $n$ is the number of items, and $\alpha()$ is the inverse Ackermann function. Due to the extremely slow growth rate of $\alpha()$, it remains less than $5$ for all possible values of $n$ that can be represented in the physical universe. 

\begin{algorithm}[t!]
\caption{Disjoint-set operators}
\label{alg:ds}
\SetKwProg{myproc}{Procedure}{}{}
\myproc{$\kw{Find}(x)$}{
    \If{$x.parent \neq x$}{
       $x.parent \leftarrow \kw{Find}(x)$\;
       \Return{$x.parent$}\;
    }
    \Return{$x$}\;
}

\SetKwProg{myproc}{Procedure}{}{}
\myproc{$\kw{Union}(x,y)$}{
    $x \leftarrow \kw{Find}(x)$\;
    $y \leftarrow \kw{Find}(y)$\;
    \lIf{$x = y$}{\Return}
    \lIf{$x.size > y.size$}{$\kw{swap}(x,y)$}
    $x.parent \leftarrow y$\;
    $y.size \leftarrow x.size + y.size$\;
}
\end{algorithm}

% call the tree structure maintained in \refalg{ds} \dstree. 
Using disjoint set is a fundamental method to detect all connected components and is also natural to deal with scenarios with only edge insertions. However, given an edge deletion disconnecting two connected components, it is hard to split a set based on the disjoint-set structure. Identifying all vertices belonging to one of the resulting connected components is challenging.
We define \dstree as the tree structure maintained in \refalg{ds}. Note that some additional attributes will be maintained for each vertex in \dstree, which are used for edge deletion and will be introduced later.
Our idea is to simultaneously maintain a \baseidxname and a \dstree for each connected component. We utilize \dstree to improve the performance of query processing and edge insertion. We update \dstree with the help of \baseidxname for edge deletion.

% \stitle{\textcolor{red}{Two Potential Baselines.}} Before presenting our techiques, we discuss two potential methods to maintain the spanning tree and the disjoint set simultaneously. The edge insertion is naturally handled by the disjoint set. The edge deletion can be divided into two cases. If a replacement edge is found, all connected components are consistent, and the disjoint set does not need to update. We focus on the case when the replacement edge cannot be found, and below are two potential baselines.

% % Now we propose two basic methods that combine the \dtree index and disjoint-set structures. For the disjoint-set structure, the insertion operation is very simple. In a delete operation, if replacement edges can be found, the disjoint-set does not need to be updated. Therefore here we focus on the deletion operation when the replacement edges can not be found.

% \begin{enumerate}
% \item \textit{\idtdsa} traverses all trees of \dtree and sets the parent of each vertex $u$ in the Disjoint-set to the root of $u$ in \dtree.
% \item \textit{\idtdsb.} Traverse the two newly generated trees after deleting the tree edges in the \dtree index, and set the parent of each vertex in the Disjoint-set to the root of the \dtree where it is located.
% \end{enumerate}

% Obviously, although these two methods can make the time complexity of the query operation reach a constant, the time complexity of the deletion operation, especially when no replacement edge can be found after deleting the tree edge, reaches $O(m)$. This is unacceptable for large-scale graph data.

We introduce the data structure of \dstree in \refsubsec{optstructure}. In \refsubsec{optoperator}, we study the details of several \dstree operations which serve in our final algorithms for edge insertions and deletions. Details of the final query processing algorithm will also be covered in \refsubsec{optoperator}. Then, we introduce our final edge insertion algorithm and edge deletion algorithm in \refsubsec{optinsert} and \refsubsec{optdelete}, respectively. The complexity analysis of these two algorithms is also covered there.

% \begin{example}
%     an example of ds tree and id tree
% \end{example}

% \stitle{Optimal Query Processing.} Recall that in our query processing algorithm for \baseidxname, nothing need to be maintained. Therefore, we can process queries based on \dstree and do nothing for \dtree. We find the root of each query vertex in \dstree and identify if they are the same. The query 

\subsection{\dstree Structure}
\label{subsec:optstructure}

% To implement the final algorithms for edge insertion and edge deletion, we discuss the data structure of \dstree in this subsection. In the next subsection, we propose the implementations of several \dstree operators in the frameworks.

\stitle{Optimal Children Maintenance in \dstree.} We discuss the data structure of \dstree in this subsection. We start by considering the attributes of each vertex in \dstree. For an edge deletion, it may require removing several vertices from \dstree, and we need to connect all children for each removed vertex back to the tree. Unlike \dtree, tree edges in \dstree may not be graph edges, which makes searching children from scratch (like Line 9 of \refalg{basedelete}) not feasible. Therefore, a data structure to maintain children of each vertex in \dstree is expected to support the following tasks.

\vspace{0.5em}
\begin{itemize}
\item Task 1. Deleting a vertex $u$ from children of $v$;
\item Task 2. Scanning all children of $u$;
\item Task 3. Inserting a vertex $u$ into the children set of $v$.
\end{itemize}
\vspace{0.5em}

The first two tasks are required when we remove $u$ from a \dstree. The last task is required when we delete the parent of $u$ or take the union of two \dstrees rooted in $u$ and $v$, respectively. A straightforward idea is to use a hash table. However, even if a large number of buckets are used to achieve a high efficiency for insertion and deletion, scanning all children in the second task is time-consuming.
To overcome this challenge, we adopt a doubly-linked list (DLL) structure to maintain all children for each vertex in \dstree. Specifically, we use a structure called \dsnode to represent each vertex in \dstree and maintain the following attributes for \dsnode.

\vspace{-0.2em}
\begin{table}[h!]
% \centering
\begin{tabular}{l l}
% \multicolumn{2}{l}{Members of \dsnode:}\\
% Members     & Description\\\hline\hline
- $id$        & // id of the corresponding vertex; \\
- $parent$    & // pointer to the parent's $\dsnode$ in $\dsname$-tree; \\
- $pre$       & // previous pointer in the DLL of the parent's \\ & // children; \\
- $next$      & // next pointer in the DLL of the parent's \\ & // children; \\
- $children$     & // start position of the DDL of children.
\end{tabular}
% \vspace{1em}
% \caption{Frequently used notations}
% \label{tab:notation}
\end{table}
%\vspace{-1em}

Note that we do not maintain the subtree size for each vertex. This is because we only use the subtree size of the \dstree root when linking two \dstrees. The subtree size of a vertex $u$ in \dstree is the same as that in \basetree if $u$ is the root in both trees.
The children attribute points to the first child in DLL. To remove a child $u$, we connect the previous child ($\dsnode.pre$) and the next child ($\dsnode.next$) in DLL. To insert a child $u$, we add $u$ to the beginning of the DLL. The following lemma holds.

\vspace{-0.5em}
\begin{lemma}
Inserting a new child or deleting an existing child for a vertex in \dstree is completed in $O(1)$ time.
\end{lemma}
\vspace{-0.5em}

For task 2, it is clear to see that all children can be scanned following the DLL, and the time only depends on the number of children, which is also optimal. 
% We will show formal pseudocode for each tasks when presenting algorithms to manipulate \dstrees later.

\subsection{\dstree Operators}
\label{subsec:optoperator}

We introduce several operators to manipulate \dstrees which are used in our final algorithms. We add a suffix DS to distinguish certain operators from those of \dtree. 
% \refalg{optinsert} and \refalg{optdelete}. 
All operators are presented in \refalg{dsoperator}, and we will show that all of them can be implemented in amortized constant time, which is optimal.

\begin{algorithm}[t!]
\caption{\dstree operators}
\label{alg:dsoperator}

\SetKwProg{myproc}{Procedure}{}{}
\myproc{$\algdsunlink(u)$}{
    
    % $w \leftarrow \dsnode(u).children.next.id$\;
    % \While{$w \neq$ the end virtual node in DLL}{
    %     $\algdslink(w,root_u)$\;
    % }

    % \tcc{remove $u$ from the children of its parent}
    $\dsnode(u).pre.next \leftarrow \dsnode(u).next$\;
    $\dsnode(u).next.pre \leftarrow \dsnode(u).pre$\;

    $\dsnode(u).parent \leftarrow \dsnode(u)$\;
    $\dsnode(u).pre \leftarrow \kw{Null}$\;
    $\dsnode(u).next \leftarrow \kw{Null}$\;
}

\vspace{0.2em}
\SetKwProg{myproc}{Procedure}{}{}
\myproc{$\algdslink(u,v)$}{
    \tcc{union without find and comparing size}
    \tcc{the input satisfies $\stsize(u) \le \stsize(v)$}
    \tcc{union two \dsname-Trees}
    $\dsnode(u).parent \leftarrow \dsnode(v)$\;

    % \tcc{remove $u$ from the existing DLL}
    % $\dsnode(u).pre.next \leftarrow \dsnode(u).next$\;
    % $\dsnode(u).next.pre \leftarrow \dsnode(u).pre$\;

    \tcc{add $u$ to the new DLL}
    $\dsnode(u).pre \leftarrow \dsnode(v).children$\;
    $\dsnode(u).next \leftarrow \dsnode(v).children.next$\;
    $\dsnode(u).pre.next \leftarrow \dsnode(u)$\;
    $\dsnode(u).next.pre \leftarrow \dsnode(u)$\;
}

\vspace{0.2em}
\SetKwProg{myproc}{Procedure}{}{}
\myproc{$\algfind(u)$}{
    \If{$\dsnode(u).parent \neq \dsnode(u)$}{
        $root \leftarrow \algfind(\dsnode(u).parent.id)$\;
        $\algdsunlink(u)$\;
        $\algdslink(u,root)$\;
        % $\dsnode(u).parent \leftarrow \algfind(\dsnode(u).parent.id)$\;
        \Return $root$\;
    }
    \Return $u$\;
}

\vspace{0.2em}
\SetKwProg{myproc}{Procedure}{}{}
\myproc{$\algdsdel(u)$}{
    \tcc{assign children of $u$ to the root}
    $root_u \leftarrow \algfind(u)$\;
    $\algdsunlink(u)$\;
    \ForEach{child $w$ of $u$ in $\dsname$-Tree}{
        $\algdsunlink(w)$\;
        $\algdslink(w,root_u)$\;
    }
    % $\dsnode(u).children \leftarrow$ an empty DLL\;
    
}

\vspace{0.2em}
\SetKwProg{myproc}{Procedure}{}{}
\myproc{$\algdsreroot(u)$}{
    % \tcc{reroot \basename-Tree}
    % $\algreroot(u)$\;
    % \tcc{reroot \dsname-Tree}
    $root_u \leftarrow \algfind(u)$\;
    $\kw{swap}(\dsnode(u),\dsnode(root_u))$\;
    $\dsnode(u).id \leftarrow u$\;
    $\dsnode(root_u).id \leftarrow root_u$\;
    
}

\end{algorithm}

\stitle{UnlinkDS and LinkDS.} We start with operators that are straightforward to be implemented. Changing the parent of a tree node often happens in manipulating \dstrees, \rr{and a typical scenario is the path compression optimization to find the root in a disjoint-set tree (i.e., all visited vertices are assigned as the children of the root).}
\algdsunlink disconnects the vertex $u$ from its parent. It removes $u$ from the children DLL of its parent. For ease of presentation, we add virtual vertices at the beginning and end of DLL respectively. In this way, the $pre$ pointer (Line 2) and the $next$ pointer (Line 3) are not \kw{Null}. 

\algdslink adds a vertex to the children DLL of the other vertex and is a simplified version of \kw{Union}. 
% A typical application is to merge two \dstrees for tree edge insertion. 
%
In our final update algorithms, we identify the root and the size of each \dstree before invoking \algdslink. Therefore, unlike the original \kw{Union} operator, we do not need to execute $\kw{Find}$ to find the root of each vertex and compare the size of two trees. We set $v$ as the parent of $u$ in Line 8 and add $u$ to the children DLL of $v$. For ease of presentation, we add virtual vertices at the beginning and end of DLL respectively. $\dsnode(v).children$ (Line 9) points to the virtual beginner of children DLL of $v$. In this way, the $pre$ pointer (Line 11) and the $next$ pointer (Line 12) are not \kw{Null}.
\rr{The strategy of \algdslink is consistent with the \kw{Union} operator of disjoint-set structure \cite{DBLP:journals/jacm/TarjanL84}. }The time complexity of both \algdsunlink and \algdslink is $O(1)$.

\stitle{FindDS for Query Processing.} Based on $\algdsunlink$ and $\algdslink$, we reorganize the \kw{Find} operator in the original disjoint set based on the attributes of \dsnode. The operator is executed to find the root of each vertex.
Recall that in our query processing algorithm for \baseidxname, nothing needs to be maintained. Therefore, we can process queries based on \dstree and do nothing for \dtree. Our final query processing algorithm is called $\algoptquery$. It invokes $\algfind$ to find the root of each vertex and identifies if they are the same. The theoretical querying time is summarized below.

\vspace{-0.5em}
\begin{theorem}
The amortized complexity of $\algoptquery$ for connectivity query processing is $O(\alpha(n))$, where $\alpha(n)$ is the inverse Ackermann function of the number of vertices $n$ in the graph.
\end{theorem}
\vspace{-0.5em}

\begin{proof}
\rr{The strategy of \algoptquery is consistent with the \kw{Find} operator of Disjoint-Set structure \cite{DBLP:journals/jacm/TarjanL84}. \algoptquery has the same time complexity as \kw{Find}.}
\end{proof}
\vspace{-0.5em}

Given that $\alpha(n) < 5$ as mentioned earlier, this makes an amortized constant time for our query processing algorithm in practice.

\stitle{Isolate.} The \algdsdel operator deletes a vertex from a \dstree. 
% Generally, given an edge deletion disconnecting a connected component, we split the original \dstree into two trees for the resulting connected components, respectively. 
A \dstree may require to be split into two trees for edge deletion. To this end, we \algdsdel all vertices belonging to the smaller \basetree from the old \dstree and union all vertices which are isolated as a new \dstree.
Instead of adding all children of $u$ to its parent in the pseudocode, we first do a path compression and find the root of the \dstree. Then, we link each child to the root by executing the \algdslink operator (Line 25). 

\vspace{-0.5em}
\begin{lemma}
\label{lem:dschildren}
The amortized time complexity of \algdsdel is $O(1)$.
\end{lemma}
\vspace{-0.5em}

\vspace{-0.5em}
\begin{proof}
\rr{The lemma is straightforward since the number of all children in the tree is $n$, and the amortized children number of each vertex is $O(1)$.}
% \algdsdel needs to traverse all children of $u$ on the DS-tree. In the DS-tree, each vertex can only be the child of one other vertex. Then the number of all children on the entire DS-tree is $n$, and the average children number of each vertex is $O(1)$. Therefore, the time complexity of \algdsdel is also $O(1)$.
\end{proof}
\vspace{-0.5em}

\rr{The $\algdsdel$ operation removes a given vertex $u$ from its connected component. In the operation, all children of $u$ are added back to the \dstree as the children of the root, which guarantees the children are not removed. Note that given the path compression optimization, the children number of the \dstree root is very large, and that for all other tree vertices is very small. 
% In other words, only the tree roots have a great number of children. 
\textbf{Our final algorithm for edge deletion will guarantee to never execute $\algdsdel$ for the root. } }

\vspace{-0.3em}
\begin{example}
An example of executing \algdsdel is shown in \reffig{dll}. Assume that $r$ is the \dstree root. To delete the vertex $u$, we first connect the previous vertex and the next vertex in the doubly linked list. Then we connect the children of $u$ to the root $r$.
\end{example}
\vspace{-0.3em}

\begin{figure}[t!]
\centering
\includegraphics[width=0.5\columnwidth]{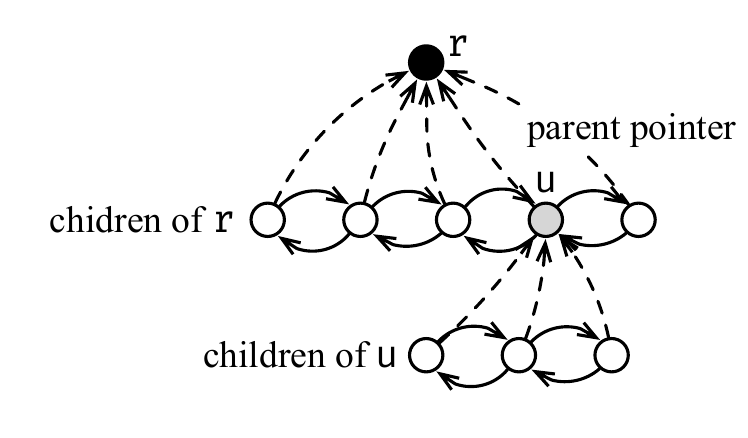}
\vspace{-0.8em}
\caption{An example of $\algdsdel(u)$.}
\label{fig:dll}
\end{figure}

\stitle{RerootDS.} The \algdsreroot operator is crucial to keep the \dstree root consistent with that of \basetree. 
% potential ways
Several potential ways exist to reroot \dstree. One way is to rotate the \dstree similar to the \algreroot operator in \basetree where all tree edges do not change. Given the new root $u$, another way is to first delete $u$ from the tree by invoking the \algdsdel operator and then assign $u$ as the parent of the original root by invoking the \algdslink operator.
However, both of the above methods increase the depth of vertices in \dstree, which reduces the efficiency of \algfind and may even break the amortized constant time complexity of \algfind.

Motivated by this, our method to reroot \dstree is to replace the new root $u$ with the original root $root_u$ directly. By replacement, we mean to exchange the children and the parent of $u$ and $root_u$. Note that the parent attribute of every child of $root_u$ will be pointed to $u$. Even if the amortized number of children for each tree vertex is $O(1)$ as proved in \reflem{dschildren}, the size of all children of $root_u$ can be very large. To improve the reroot efficiency, we store a \dsnode pointer instead of a vertex id in the parent attribute for each \dsnode. In this way, as shown in Lines 28--30 of \refalg{dsoperator}, we swap the corresponding \dsnode objects of two vertices, and update their id for the new vertex. The time complexity of \algdsreroot is $O(1)$.

\subsection{The Final Edge Insertion Algorithm}
\label{subsec:optinsert}

For edge insertion, we \kw{Union} two \dstrees if a tree edge is inserted to connect two connected components. Additionally, we propose an optimization to improve the performance of both insertion and deletion. The optimization keeps the root of \dstree the same as that of \basetree. It benefits the following operations in edge insertion and deletion, respectively.

\vspace{0.5em}
\begin{enumerate}
\item \textit{Avoid searching \basetree root.} 
As shown in \refalg{baseinsert}, we need to search the root of \basetree frequently. Given that finding the \basetree root takes amortized constant time, keeping their roots consistent speeds up the edge insertion significantly.

% \item \textit{Avoid maintaining \dstree node size.} 
% We observe that the subtree size of each node in \dstree is only considered when the node is root in \kw{Union}.

\item \textit{Avoid deleting \basetree root.}
In Line 3 of \refalg{basedelete}, we split the subtree $T_u$ of $u$ from the original tree if a replacement edge is not found. In our final edge deletion framework discussed later, we will delete all vertices in $T_u$ from \dstree. Keeping the two roots of \basetree and \dstree consistent guarantees that the \dstree root is never deleted. An immediate challenge after deleting the \dstree root is to identify a new root. Then we need to link all children of $u$ back to the new root. Due to the path compression optimization, the number of children is extremely large. Connecting all children of the \dstree root to the new root is costly. 
\end{enumerate}
\vspace{0.5em}

\begin{algorithm}[t!]
\caption{\algoptinsert}
\label{alg:optinsert}
\KwIn{an existing edge $(u,v)$ and the \optidxname index}
\KwOut{the updated \optidxname}

$root_u \leftarrow \algfind(u)$\;
$root_v \leftarrow \algfind(v)$\;
Lines 3--14 of \refalg{baseinsert}\;
$\algreroot(u)$\;
% $\algdsreroot(u)$\;
$\algdslink(root_u,root_v)$\;
$\alglink(u,v,root_v)$\;
\end{algorithm}

\begin{algorithm}[t!]
\caption{\algoptdelete}
\label{alg:optdelete}
\KwIn{an existing edge $(u,v)$ and the \optidxname index}
\KwOut{the updated \optidxname}
$(u,root_v,succ,S) \leftarrow \algbasedelete(u,v)$\;
$\algdsreroot(root_v)$\;
\lIf{$succ$}{\Return}
\tcc{$u$ is the root of the smaller \basename-Tree}

$\algdsdel(u)$\;
$S \leftarrow S \setminus \{u\}$\;
\ForEach{$w \in S$}{
    $\algdsdel(w)$\;
    $\algdslink(w,u)$\;
}

\end{algorithm}

\noindent We present the final algorithm for edge insertion in \refalg{optinsert}. The final index is called \optidxname (short for I\underline{D}-Tree a\underline{n}d \underline{D}S-Tree), and the insertion algorithm is called $\algoptinsert$. 
%
% Both our final insertion algorithm in this section and the deletion algorithm in \refsubsec{optdelete} involve several \dstree operators. We add a suffix DS to distinguish certain operators from those of \dtree. The efficiency of our final algorithms depends on how we implement these \dstree operators. We will introduce details in \refsubsec{optoperator}.
%
In Lines 1--2, we find the root of each vertex in \dstree by invoking \algfind. Given that roots of two tree are consistent, fining root in \dstree is much more efficient. 
For non-tree edge insertion, we reorganize the \basetree same as \refalg{baseinsert}. Given that the tree root does not change, we do nothing for \dstree.
For tree edge insertion, we need \kw{Union} two \dstrees given that two connected components are connected. After Line 14 of \refalg{baseinsert}, we already know $root_u$ and $root_v$ are roots of two original \dstrees and the subtree size of $root_u$ is smaller than that of $root_v$. Note that in Line 4, we \algreroot the \basetree of $root_u$ to $u$ but do not keep the root of \dstree consistent. This would not break the correctness since the tree will be merged into the larger tree, and the root of the final \dstree will be correct.
In Line 5, we invoke \algdslink to assign $root_u$ as the child of $root_v$ in \dstrees which essentially merges two \dstrees.
Finally, we link \basetrees as before.
Given the implementation details of all \dstree operators, we have the following theorem.

% it is clear to see that the time complexity of our final update algorithms never increases when the time complexity of query processing is reduced to optimal.

\vspace{-0.5em}
\begin{theorem}
\label{thm:finalinsert}
%The running time complexity of \refalg{optinsert} is bounded by that of \refalg{baseinsert}.
\rr{\refalg{optinsert} has the same running time complexity as \refalg{baseinsert}.}
\end{theorem}
\vspace{-0.8em}

\subsection{The Final Edge Deletion Algorithm}
\label{subsec:optdelete}

Recall that in \refalg{basedelete} for edge deletion, we first \kw{unlink} the \basetree. After Line 4 of \refalg{basedelete}, $u$ is the root of the smaller tree for searching replacement edge, and $root_v$ is the root of the larger tree. 
To update the corresponding \dstree, we still consider two cases based on the existence of the replacement edge. If a replacement edge is found, the connected component does not update, and the tree root is $root_v$ in Line 23 of \refalg{basedelete}. We only reroot the \dstree to $root_v$ in this case.
If a replacement edge is not found, we need to split the original \dstree into two trees for two resulting \basetrees. To this end, our idea is to delete every vertex belonging to the subtree of $u$ in \basetree from the \dstree. Then, we \kw{Union} all deleted vertices and form the new \dstree. An immediate challenge is how to identify all vertices belonging to the subtree of $u$ in \basetree. We discuss this in the following lemma.

\begin{lemma}
\label{lem:subtree}
Given $u$ in Line 5 of \refalg{basedelete}, if no replacement edge is found (i.e., $succ$ turns false), $S$ is the set of all vertices in the subtree rooted in $u$ when the algorithm terminates.
\end{lemma}

% \begin{proof}
%     We start BFS from $u$. When traversing $x$'s neighbor $y$, only when $(x, y)$ is a tree edge will $y$ be added to the queue. If no replacement edge can be found, all tree edges on the subtree of $u$ are visited. Finally $S$ will contain all vertices on the subtree of $u$.
% \end{proof}

Based on \reflem{subtree}, we present our final deletion algorithm, called $\algoptdelete$, for \optidxname index in \refalg{optdelete}. After updating \basetrees, we first reroot the \dstree to $root_v$ in Line 2. Given $succ = \kw{true}$ in Line 3, we terminate the algorithm since the connected component does not update. Otherwise, we delete all vertices in $S$ from the \dstree in Lines 4--8. 
% $\algdsdel$ is an operator to delete a vertex from a \dstree. 
We first delete $u$ since $u$ will be the root of the new \dstree. Then, we iteratively delete each vertex $w$ from $S$ and link $w$ to $u$. \rr{$S$ contains all vertices in the subtree of $u$ in the ID-tree. Given that the parent of $u$ is $root_v$ (Line 1), $root_v$ is not in $S$. In Line 2 of \refalg{optdelete}, we update the root of the disjoint set tree to $root_v$. Therefore, we never \algdsdel the root of the disjoint-set tree.} From the perspective of disjoint set, we perform the \kw{Union} operation on $u$ and $w$ in Line 8. As introduced in \refsubsec{optinsert}, $\algdslink(w,u)$ directly assigns $w$ as a child of $u$ since the subtree size of $u$ is guaranteed larger than that of $w$ in \dstree and both of them are roots before $\kw{Union}$.

\begin{theorem}
\label{thm:finaldelete}
%The running time complexity of \refalg{optdelete} is bounded by that of \refalg{basedelete}.
\rr{\refalg{optdelete} has the same expected running time complexity as \refalg{basedelete}.}
\end{theorem}
\vspace{-1em}
\begin{proof}
On the basis of \refalg{basedelete}, \refalg{optdelete} also needs to traverse all the vertices in $S$ to decompose the DS-tree. Since all DS-Tree operations are $O(1)$, the added operation is bounded by $|S|$. \refalg{basedelete} is also bounded by $|S|$, so the complexity of the two deletion operations is still the same.
\end{proof}

%!TEX root = main.tex

\section{2-Edge Connectivity Maintenance}
\label{sec:2ecc}

We extend our connectivity maintenance algorithms to maintain 2-edge connectivity between vertices. Two vertices are 2-edge connected if there exist at least two edge-disjoint paths between them, meaning that the removal of any single edge would not disconnect the two vertices. Maintaining 2-edge connectivity is important because it captures the robustness of a network against edge failures and is widely used in applications such as network reliability analysis \cite{krekovic2025reducing}, biological analysis \cite{chaudhuri2022network}. 
A 2-edge connected component in a graph is a maximal subgraph in which any two vertices are 2-edge connected. Two vertices are 2-edge connected if and only if they are in the same 2-edge connected component.
In this section, we propose a spanning-tree-based index to maintain 2-edge connected components. Section 8 introduces several basic operations for index maintenance. Section 9 presents a method that uses the Disjoint-Set Structure to further improve query efficiency to constant theoretical time.

\subsection{Maintaining 2-Edge Connected Components based on Spanning Trees}
\label{subsec:motivation}

In this section, we propose an index based on spanning trees for maintaining 2-edge connected components. 
%Section 8 details the update operations for our index. Section 9 presents a method utilising a disjoint-set structure to further enhance query efficiency to constant time. 
We first study the relationship between connected component and 2-edge connected component.

\begin{lemma}
\label{lem:cc2cc}
Given a connected component $C$, a spanning tree $T$ of $C$, each 2-edge connected component of $C$ corresponds to a connected substructure of $T$.
\end{lemma}

\begin{proof}
    A bridge is an edge whose removal disconnects the graph, and a subgraph is 2-edge connected if it remains connected after the removal of any single edge. A 2-edge connected graph has no bridge. Since every bridge must appear as a tree edge in the spanning tree, every 2-edge connected component in the spanning tree must be connected. 
\end{proof}

\begin{example}
    A possible spanning tree and non-tree edges of
Figure 1 are presented in Figure 2. Solid lines indicate tree edges in the spanning tree and dashed lines indicate non-tree edges that
are not in the spanning tree. We can see that for the three 2-edge connected components in Figure 1, each 2-edge connected component can correspond to a connected substructure in the spanning tree.
\end{example}

% \begin{definition}
% \label{def:2ID}
%     Given a connected component $C$, a spanning
% tree $T$ of $C$, and a 2-edge connected component $TC$ in $C$, the \basetidxname of $TC$ is the connected substructure of $T$ that corresponds to $TC$. 
% \end{definition}

%Based on \refdef{2ID}, the 2-edge connected component maintenance problem can be converted to maintaining \basetidxname. The problem of biconnectivity of two given query vertices can be converted to checking whether they are in the same \basetidxname. 

Based on \reflem{cc2cc}, we can propose a spanning-tree-based index to maintain 2-edge connectivity. If two vertices $u$ and $v$ in a spanning tree are 2-edge connected, all adjacent vertices along the tree path between $u$ and $v$ are 2-edge connected. In particular, connected components can be maintained by spanning trees, and 2-edge connected components can be maintained by maintaining certain auxiliary structures in spanning trees. The structures are motivated by the following lemmas.
%We first show how to identify 2-edge connectivity for adjacent vertices in the spanning trees. 

% In addition to \baseidxname, we maintain 2-edge connectivity for adjacent vertices in the spanning trees supported by the following lemma.

\begin{lemma}
\label{lem:rep}
Given a connected component $C$, a spanning
tree $T$ of $C$ and a tree edge $(u,v)$ in $T$, $u$ and $v$ are in the same 2-edge connected component if and only if there exists at least one replacement edge for $(u,v)$.
\end{lemma}

\begin{proof}
    If a tree edge $(u,v)$ has a replacement edge $(x,y)$, there must exist a path containing $(x,y)$ that connects $u$ and $v$. If there exists a path connecting $u$ and $v$ that does not contain $(u,v)$, this path must contain a non-tree edge that can serve as a replacement for $(u,v)$.
\end{proof}

% Then, we identify 2-edge connectivity for non-adjacent vertices in the spanning trees supported by the following lemma.

\begin{lemma}
\label{lem:2ecctrans}
Given three verices $u,v$ and $w$, if $u$ and $v$ are 2-edge connected, and $v$ and $w$ are also 2-edge connected, $u,v$ and $w$ are in the same 2-edge connected component \cite{DBLP:journals/jacm/HolmLT01}. 
\end{lemma}

Based on \reflem{rep}, if we know the number of replacement edges of each tree edge, we can identify the 2-edge connectivity between the two vertices of the tree edge. 
Based on \reflem{2ecctrans}, we can deduce the 2-edge connectivity between any two vertices if we can maintain the 2-edge connectivity between the two vertices of every tree edge.
Below, we formally present our first index for 2-edge connectivity, named \basetidxname.
Compared to \baseidxname, \basetidxname maintains an additional attribute for each vertex $u$ as follows.

\begin{algorithm}[t!]
\caption{Basic Operations for \basetidxname}
\label{alg:2etree}

\proc{$\alglink(u,v)$}{
    $u.\parent \gets v$\;
    $u.rep \gets 0$\;

    $v.\stsize \gets v.\stsize + u.\stsize$\;
    
    %\Return $v$\;
}

\proc{$\algcut(u,v)$}{
    $u.\parent \gets \kwnull$\;
    % $u.\orep \gets 0$,
    % $u.\sn \gets \emptyset$\;
    \While{$v \neq \kwnull$}{
        $v.\stsize \gets v.\stsize - u.\stsize$\;
        $v \gets v.\parent$\;
    }
}

\proc{$\algreroot(u)$}{
    \lIf{$u.\parent = \kwnull$}{\Return}
    $v \gets u.\parent$\;
    $\algreroot(v)$\;

    $v.\parent \gets u$,
    $u.\parent \gets \kwnull$\;
    %$v.rep \gets 0$;
    $swap(u.rep,v.rep)$\;

    $v.\stsize \gets v.\stsize - u.\stsize$\;
    $u.\stsize \gets v.\stsize + u.\stsize$\;
    
}

\end{algorithm}

\begin{algorithm}[t!]
\caption{\algtquery}
\label{alg:2equery}
% \SetAlgoVlined
\KwIn{two query vertices $u$ and $v$}
\KwOut{the 2-edge connectivity between $u$ and $v$}

    $root_u \leftarrow GetC^{2}root(u)$\;
    $root_v \leftarrow GetC^{2}root(v)$\;
    \Return $root_u = root_v$\;
\proc{$GetC^{2}root(u)$}{

    \While{$u.rep \neq 0$}{
        $u \gets u.\parent$\;
        
    }   
    \Return $u$\;
    }

\end{algorithm}

\begin{itemize}
\item $rep(u)$: the number of replacement edges of the tree edge $(u,u.parent)$. If $u$ is the root, $u.rep=0$.

\end{itemize}

Given the additional label $rep$ stored for each vertex compared with \baseidxname, we present several basic tree operators for \basetidxname in \refalg{2etree}. 

\stitle{Link and Cut-Bridge.} Link connects the roots of two trees and updates the subtree size accordingly. The vertex $v$ is the root of the merged tree. We set $u.rep = 0$ because there is no replacement edge for the new tree edge $(u, v)$. Cut-Bridge deletes a bridge edge $(u, v)$ from the tree where $v$ is the parent of $u$. $u$ will be the root of one tree. The time complexity of Link and Cut-Bridge is $O(1)$.

\stitle{ReRoot.} ReRoot rotates the tree to make $u$ as the tree root. During rotation, we also correspondingly update $\stsize$ and $rep$ (Lines 15--17). The time complexity of ReRoot is $O(h)$.

\subsection{Query Processing Based on \basetidxname}
\label{subsec:2ID-tree query}

Given two query vertices, our idea for 2-edge connectivity testing is to identify if they are in the same \basetidxname. To this end, we define the concept of C$^{2}$root as follows.

\begin{definition}
\label{def:2IDroot}
    Given a connected component $C$, a spanning
tree $T$ of $C$, and a 2-edge connected component $TC$, the C$^{2}$root $r$ is the vertex with the lowest depth in the \basetidxname of $TC$. 
\end{definition}

For each vertex $u$ in $TC$, we can also call $r$ in \refdef{2IDroot} as the C$^{2}$root of $u$. Given two query vertices, they have the same C$^{2}$root if and only if they are in the same 2-edge connected component. 

%\xx{change ID$^{2}$root to a different name? It is now feel like the root of ID$^2$Tree}

\begin{lemma}
\label{lem:2IDroot}
    Given a connected component $C$, a spanning
tree $T$ of $C$, and a 2-edge connected component $TC$, if $r$ is the C$^{2}$root of $TC$, either $r$ is the root of $T$, or the tree edge $(r, r.\parent)$ has no replacement edge.
\end{lemma}

\begin{proof}
    Based on \refdef{2IDroot}, the root of the spanning tree can be an C$^{2}$root. If $r$ is not the root of the spanning tree and the tree edge $(r, r.\parent)$ has a replacement edge, $r.\parent$ can be the vertex with the lower depth in $TC$. Therefore, $r$ is a C$^{2}$root if $(r, r.\parent)$ has no replacement edge.
\end{proof}

Based on \reflem{2IDroot}, we present our query algorithm in \refalg{2equery}. We determine whether two query vertices are in the same 2-edge connected component by comparing their C$^{2}$roots (Lines 1--3). To get the C$^{2}$root of a vertex $u$, we need to traverse the ancestors of $u$ in descending order of depth, where the C$^{2}$root of $u$ is the first vertex with $rep=0$ (Lines 4--7). The time complexity of the query algorithm is $O(h)$.

\section{Update Processing Based on \basetidxname}

The updating procedure for the spanning tree structure in ID$^{2}$Tree is the same as \baseidxname. Therefore, we mainly focus on the maintenance of the replacement edge count. 

\begin{lemma}
\label{lem:nontree}
Given the spanning tree $T$ for a graph $G$, let $T'$ be the spanning tree for the graph $G'$ by inserting or deleting an edge from $G$. The replacement edge of every tree edge in $T$ does not change if the non-tree edges do not change, i.e., $G \setminus T = G' \setminus T'$.
\end{lemma}

\begin{proof}
$G \setminus T = G' \setminus T'$ happens in two cases. The first case is that a new edge is inserted, and the edge connects two existing trees. The second case is that a tree edge is deleted, and no existing edge can reconnect the trees. In both cases, the path in the tree between two terminals of any non-tree edge does not change. Therefore, the replacement edge does not update.
\end{proof}

\begin{algorithm}[t!]
\caption{\algteinsert}
\label{alg:2einsert}
% \SetAlgoVlined
\KwIn{a new edge $(u,v)$ and the \basetidxname index}
\KwOut{the updated \basetidxname}

$root_u \leftarrow$ compute the root of $u$\;
$root_v \leftarrow$ compute the root of $v$\;

\tcc{non-tree edge insertion}
\If{$root_u = root_v$}{
$fu \gets u, fv \gets v$\;
    \While{$fu \neq fv$}{
    \uIf{$fu.\stsize < fv.\stsize$}{
    $fu.rep\gets fu.rep+1$\;
    $fu \gets fu.parent$\;
    }
    \Else{
    $fv.rep\gets fv.rep+1$\;
    $fv \gets fv.parent$\;
    }
    
}

}
\tcc{tree edge insertion}
Lines 12--15 of \refalg{baseinsert}\;

\end{algorithm}

\begin{algorithm}[t!]
\caption{\algtdeletent}
\label{alg:2edeletent}
% \SetAlgoVlined
\KwIn{a non-tree edge $(u,v)$ and the \basetidxname index}
\KwOut{the updated \basetidxname}

    \While{$u \neq v$}{
    \uIf{$u.\stsize < v.\stsize$}{
    $u.rep\gets u.rep-1$\;
    $u \gets u.parent$\;
    }
    \Else{
    $v.rep\gets v.rep-1$\;
    $v \gets v.parent$\;
    }
    
}

\end{algorithm}

\reflem{nontree} shows that updates of replacement edges are mainly caused by the update of non-tree edges, which includes the following two cases.

% Next, we will discuss the cases where replacement edges need to be updated.

\begin{itemize}

\item\stitle{Insert or delete a non-tree edge.} Inserting a non-tree edge $(u,v)$ will add a replacement edge for all the edges in the tree path between $u$ and $v$. Deleting a non-tree edge $(u,v)$ will remove a replacement edge for all the edges in the tree path between $u$ and $v$. Edges outside path between $u$ and $v$ in the spanning tree are not affected.

\item\stitle{Delete a tree edge with replacement edges.} Cutting the tree edge $(u,v)$ causes the spanning tree to be split into two trees. We need to reconnect the two trees by changing one of the replacement edges of $(u,v)$ from a non-tree edge to a tree edge.

\end{itemize}

\subsection{ID$^{2}$Tree Maintenance for Insertion}

In this section, we propose edge insertion algorithm for maintaining the ID$^{2}$Tree in \refalg{2einsert}. 
To insert a new tree edge $(u,v)$, it will connect two spanning trees, and no replacement edge of any tree edge would change. Therefore, tree edge insertion of \basetidxname is the same as \baseidxname (Line 12 in \refalg{2einsert}). 

We next discuss the insertion of non-tree edges. The new non-tree edge $(u,v)$ can be a replacement edge for all edges in the tree path between $u$ and $v$ in the tree. The $rep$ of every edge in the path increases by 1 (Lines 3--11 in \refalg{2einsert}).

\begin{lemma}
    The time complexity of \refalg{2einsert} is $O(h)$.
\end{lemma}

\begin{proof}
    The additional maintenance of $rep$ compared to the \baseidxname would not change the complexity.
\end{proof}

It is worth noting that our non-tree edge insertion algorithm does not apply the BFS heuristic shown in lines 3--11 of \refalg{baseinsert}. The BFS heuristic can reduce the average search depth but does not increase the time complexity of insertion operations. However, during the maintenance of a \basetidxname, this involves the deletion of a tree edge. The number of replacement edges for more edges can be affected, and maintaining these replacement numbers is time-consuming. The application of the BFS heuristic will increase the complexity of the insertion algorithm. We will discuss this in detail in our edge deletion algorithm.

\begin{algorithm}[t!]
\caption{\algtdeletet}
\label{alg:2edeletet}
% \SetAlgoVlined
\KwIn{a tree edge $(u,v)$ and the \basetidxname index}
\KwOut{the updated \basetidxname}

\lIf{$v.parent = u$}{$\kw{swap}(u,v)$}
$ru\gets u$,$rv\gets v$\;
\While{$rv.parent \neq \kwnull$}{
    $rv \gets rv.parent$\;
}
\If{$ru.\stsize > rv.\stsize - ru.\stsize$}{$\algreroot(ru),\kw{swap}(ru,rv)$\;}
$P \gets Getrep(ru,rv)$\;

\ForEach{$(x,y) \in P$}{
        $\algtdeletent(x,y)$\;
        }
$\algcut(u,v)$\;

\ForEach{$(x,y) \in P$}{
        $\algteinsert(x,y)$\;
        }

\proc{$Getrep(u,v)$}{
    $P \gets$ an empty edge set\;
    $Q \gets$ an empty queue\;
    $S \gets$ an empty set\;
    
        $Q.push(u)$,
        $S.insert(u)$\;
        \While{$Q \neq \emptyset$}{
        $x \gets Q.top()$,
        $Q.pop()$\;
        \ForEach{$y \in \nbr(x)$}{
        \If{$y.parent = x \land y \neq v$}{
        $Q.push(y)$,
        $S.insert(y)$\;
        
        }
        \Else{
        $P.insert(x,y)$\;
        }

        }
        }
        \ForEach{$(x,y) \in P$}{
        \If{$y \in S$}{
            $P.erase((x,y))$\;
        }
        }

    \Return $P$\;
    }

\end{algorithm}

\subsection{ID$^{2}$Tree Maintenance for Deletion}

In this section, we propose the algorithm of
maintaining the ID$^{2}$Tree for edge deletion. We will
discuss the deletion of non-tree edges and tree edges separately.

\stitle{Non-tree edge deletion.} The existing non-tree edge $(u,v)$ is a replacement edge for all edges in tree path between $u$ and $v$.
%\xx{$(u,v)$ ((u,v) represents an edge, not a path. Correct this presentation in all other places)}. 
The $rep$ of all edges in tree path between $u$ and $v$ decreases by 1. We propose our non-tree edge deletion algorithm in \refalg{2edeletent}. The time complexity of \refalg{2edeletent} is $O(h)$.

\stitle{Tree edge deletion.} If a tree edge $(u,v)$ is deleted, the spanning tree will be split into two trees. We need to find a replacement edge to reconnect the two trees. If the $rep$ of the tree edge $(u,v)$ is 0, we can just cut it and no other edge's $rep$ will change. If the $rep$ of the tree edge $(u,v)$ is not 0, the $rep$ of some other tree edges may change during this processing. To ensure all $rep$ can be correctly maintained, our idea is to first delete all replacement edges of the tree edge $(u,v)$, transforming $(u,v)$ into a tree edge without replacement edges, then cut $(u,v)$, and finally reinsert all deleted replacement edges. When we reinsert these replacement edges, the first edge will be inserted as a tree edge, and the other edges will be inserted as a non-tree edge. 

\begin{algorithm}[t!]
\caption{DS$^{2}$Tree operators}
\label{alg:2dsoperator}

\SetKwProg{myproc}{Procedure}{}{}
\myproc{$\algdslink(u,v)$}{

    $root_u \gets \algfind(u)$\;
    $root_v \gets \algfind(v)$\;

    \If{$\dsnode(root_v).dsize \le \dsnode(root_u).dsize$}{
    $swap(root_u,root_v)$\;
    }
    $root_u.dsize \gets root_u.dsize+root_u.dsize$\;
    Lines 8--12 of \refalg{dsoperator}\;
    
}

\vspace{0.2em}
\SetKwProg{myproc}{Procedure}{}{}
\myproc{$\algdsdel(u)$}{
    Lines 19--23 of \refalg{dsoperator}\;
    $root_u.dsize \gets root_u.dsize-1$\;
    $u.dsize \gets 1$\;
    
}

\vspace{0.2em}
\SetKwProg{myproc}{Procedure}{}{}
\myproc{$\algdsreroot(u)$}{
    Lines 27--30 of \refalg{dsoperator}\;
    $u.dsize \gets root_u.dsize$\;
    
}

\end{algorithm}

We propose our tree edge deletion algorithm in \refalg{2edeletet}. We need to search the replacement edges of the tree edge.  Similar to \baseidxname, we always search from the
smaller tree (Lines 1--6). We first delete all the replacement edges (Lines 8--9) and cut $(u,v)$ as a tree edge without replacement edges (Line 10). Finally, we reinsert all the replacement edges (Lines 11--12). $Getrep$ is to compute all the replacement edges of $(u,v)$. We store all vertices of the small tree in $S$, and all non-tree edges containing these vertices in $P$ (Lines 18--24). We traverse the edges in $P$, removing all edges where both vertices lie on the small tree (Lines 25--27). Only edges where one vertex in the small tree and the other in the big tree remain in $P$ (i.e., all replacement edges).

\begin{lemma}
    The time complexity of \refalg{2edeletet} is $O(h^{2}d)$, where $d$ is average degree of a vertex.
\end{lemma}

\begin{proof}
    The sum of the subtree size of all vertices in a tree is equal to the sum of the depth of all vertices. Therefore, the average size of the subtree of a vertex can be considered as $O(h)$. The size of $P$ in Line 7 is $O(hd)$. All edges in $P$ will be deleted by \refalg{2edeletent} and reinsert by \refalg{2einsert}. Therefore, the time complexity of \refalg{2edeletet} is $O(h^{2}d)$.
\end{proof}

\vspace{-0.6cm}
\section{Improving 2-Edge Connectivity Query Efficiency with Disjoint-Set Structure}

In the DND-tree introduced in \refsec{opt}, we maintain a disjoint-set structure for each connected component, enabling queries to be processed in almost constant time. To achieve almost constant-time queries for 2-edge connectivity, we propose a new approach \opttidxname. \opttidxname extends \basetidxname by maintaining an additional disjoint-set structure, called the DS$^{2}$Tree, for each 2-edge connected component. To be specific, \opttidxname = \basetidxname + DS$^{2}$Tree.

% \begin{lemma}
%     If a edge is inserted or deleted, there are two cases when the 2-edge connected components be changed.

%     \begin{enumerate}
%     \item 
% \end{enumerate}
% \end{lemma}

\subsection{DS$^{2}$Tree structure}

Similar to \dstree, we use a disjoint-set structure to maintain each 2-edge connected component, and the disjoint-set structure supports constant-time queries. We propose DS$^{2}$Tree as the disjoint-set structure in our algorithm for 2-edge connectivity. It is worth noting that during updates, we do not maintain consistency between the roots of the \basetidxname and the DS$^{2}$Tree (we will discuss the reason in \refsubsec{tdelete}). Therefore, the root of a DS$^{2}$Tree needs to maintain the size of the DS$^{2}$Tree. Different from \dstree, the DS$^{2}$Tree maintains the following additional attribute for each tree vertex.

\begin{table}[h]
% \centering
\begin{tabular}{l l}

- $dsize$     & // If this vertex is the root of the DS$^{2}$Tree, \\ & // $dsize$ is the number of vertices in the DS$^{2}$Tree.
\end{tabular}
% \vspace{1em}
% \caption{Frequently used notations}
% \label{tab:notation}
\end{table}

We show how to maintain $dsize$ in \refalg{2dsoperator}. In LinkDS of \refalg{2dsoperator}, we ensure that the smaller trees are linked to the larger trees (Lines 4--5). In Isolate and ReRoot, we change $dsize$ of the root of the DS$^{2}$Tree (Lines 10--11 and Line 14). The time complexity of the query in DS$^{2}$Tree can reach $O(1)$.

\subsection{\opttidxname Maintenance for Insertion}

We first discuss the case that the 2-edge connecticity of two vertices of a tree edge changes. 

\begin{lemma}
\label{lem:insrep}
    Given a tree edge $(u,v)$, the two 2-edge connected components that respectively contain $u$ and $v$ can be merged if and only if the $rep$ of the tree edge $(u, v)$ changes from 0 to 1.
\end{lemma}

\begin{proof}
    If $(u, v)$ is a newly inserted tree edge, it is clear that no 2-edge connected components will be changed. If the $rep$ of $(u, v)$ changes from $x$ to $x+1(x>0)$, $u$ and $v$ are already 2-edge connected. According to \reflem{rep}, when $(u, v)$ gets a replacement edge, $u$ and $v$ become 2-edge connected. We can union the two 2-edge connected components containing $u$ and $v$ respectively.
\end{proof}

When a non-tree edge $(u,v)$ is inserted, the $rep$ of each edge in the tree path between $u$ and $v$ increases by one. Based on \reflem{insrep}, if the $rep$ of a tree edge $(x,y)$ in the tree path between $u$ and $v$ changes from 0 to 1, the two 2-edge connected components that respectively contain $u$ and $v$ can be merged. 

We propose our edge insertion algorithm of \opttidxname in \refalg{2edinsert}. If $u$ and $v$ are in the same spanning tree, we traverse all edges in the tree path between $u$ and $v$. If the $rep$ of an edge changes from 0 to 1, we use LinkDS to merge the two 2-edge-connected components that respectively contain its vertices (Lines 6--8 and 11--13). If $u$ and $v$ are not in the same spanning tree, \refalg{2edinsert} is the same as \refalg{2einsert}.

\begin{lemma}
    The time complexity of \refalg{2edinsert} is $O(h)$.
\end{lemma}

\begin{proof}
    Compared with \refalg{2einsert}, \refalg{2edinsert} use LinkDS to merge two 2-edge-connected components. LinkDS can be done in $O(1)$. Therefore, The time complexity of \refalg{2edinsert} is the same as \refalg{2einsert}.
\end{proof}

\begin{algorithm}[t!]
\caption{\algtedinsert}
\label{alg:2edinsert}
% \SetAlgoVlined
\KwIn{a new edge $(u,v)$ and the \basetidxname index}
\KwOut{the updated \basetidxname}

Lines 27--30 of \refalg{2einsert}\;

\tcc{non-tree edge insertion}
\If{$root_u = root_v$}{
$fu \gets u, fv \gets v$\;
    \While{$fu \neq fv$}{
    \uIf{$fu.\stsize < fv.\stsize$}{
    $fu.rep\gets fu.rep+1$\;
    \If{$fu.rep=1$}{
    $\algdslink(fu,fu.parent)$\;
    }
    $fu \gets fu.parent$\;
    }
    \Else{
    $fv.rep\gets fv.rep+1$\;
    \If{$fv.rep=1$}{
    $\algdslink(fv,fv.parent)$\;
    }
    $fv \gets fv.parent$\;
    }
    
}

}

Lines 12--15 of \refalg{baseinsert}\;

\end{algorithm}

\subsection{\opttidxname Maintenance for Deletion}
\label{subsec:tdelete}

In this section, we propose our edge deletion algorithm for \opttidxname. We will
discuss the deletion of non-tree edges and tree edges separately. 

\begin{lemma}
\label{lem:delrep}
    Given a tree edge $(u,u.\parent)$ and a non-tree edge to be deleted, the two 2-edge connected component $ECC$ that contain $u$ and $v$ can be split if and only if the $rep$ of the tree edge $(u, u.\parent)$ changes from 1 to 0. 
\end{lemma}

\begin{proof}
    After a non-tree edge is deleted, if the $rep$ of $(u, u.\parent)$ is still greater than 1, $u$ and $u.\parent$ remain in the same 2-edge-connected component. If the $rep$ of $(u, u.\parent)$ changes from 1 to 0, $u$ and $u.\parent$ are not 2-edge connected now. $u$ and the descendants of $u$ that are in the $ECC$ will be removed from the $ECC$ and form a new 2-edge connected component.
\end{proof}

\stitle{Non-tree edge deletion.} We let vertex $x$ be the the lowest common ancestor of $u$ and $v$. We traverse all edges in the tree path between $u$ and $x$ and the path between $v$ and $x$ in descending order of depth. If the $rep$ of an edge in the tree path between $u$ (or $v$) and $x$ changes from 0 to 1, we need to find $u$ (or $v$) and all the descendants of $u$ (or $v$) that are 2-edge connected to $u$ (or $v$). All these vertices will be removed from the current $ECC$ and form a new 2-edge-connected component. 

We do not maintain consistent roots between the ID$^{2}$Tree and the DS$^{2}$Tree. We will discuss the motivation for this strategy at the end of this section. Because of this, a problem may occur when we delete $u$ and its descendants from the current DS$^{2}$Tree. The current root of the DS$^{2}$Tree may be one of the vertices that need to be deleted. To avoid this, we need to select another vertex in the DS$^{2}$Tree, outside of $u$ and its descendants, to serve as the new root of the DS$^{2}$Tree. We traverse the edges in the tree path in order of decreasing depth. $Split2ECC(u)$ is called when the tree edge $(u, u.parent)$ in the tree path has its $rep$ changed from 1 to 0. $u.parent$ must be in the current DS$^{2}$Tree but not among $u$ or its descendants. Therefore, $u.parent$ is a suitable vertex to become the new root of the DS$^{2}$Tree. For implementation convenience, in $Split2ECC(u)$, we just let $u.parent$ as the new root of the DS$^{2}$Tree.

\begin{algorithm}[t!]
\caption{\algtddeletent}
\label{alg:2eddeletent}
% \SetAlgoVlined
\KwIn{an existing edge $(u,v)$ and the \opttidxname index}
\KwOut{the updated \opttidxname}

    \While{$u \neq v$}{
    \uIf{$u.\stsize < v.\stsize$}{
    $u.rep\gets u.rep-1$\;
    \If{$u.rep=0$}{
    $Split2ECC(u)$\;
    }
    $u \gets u.parent$\;
    }
    \Else{
    $v.rep\gets v.rep-1$\;
    \If{$v.rep=0$}{
    $Split2ECC(v)$\;
    }
    $v \gets v.parent$\;
    }
    
}

\proc{$Split2ECC(u)$}{

$\algdsreroot(u.parent)$\;

    $Q \gets$ an empty queue\;

        $Q.push(u)$\;
        $\algdsdel(u)$\;

        \While{$Q \neq \emptyset$}{
        $x \gets Q.top()$,
        $Q.pop()$\;
        \ForEach{$y \in \nbr(x)$}{
        \If{$y.parent = x \land y.rep>0$}{
        $Q.push(y)$\;
        $\algdsdel(y)$\;
        $\algdslink(y,u)$\;
        
        }     
        }
        }

    }

\end{algorithm}

We propose our non-tree edge deletion algorithm of \opttidxname in \refalg{2eddeletent}. If the $rep$ of an edge changes from 1 to 0, we use $Split2ECC$ to split the current 2-edge-connected component (Lines 3--5 and 8--10). In $Split2ECC$, we make $u.\parent$ as the new root of the DS$^{2}$Tree to avoid deleting the root (Line 13). We perform a BFS to find all the descendants of $u$ that are 2-edge connected to $u$ (Lines 17--23).  For each qualified vertex, we use $isolate$ to remove it from the current DS$^{2}$Tree, and then use LinkDS to add it to the new DS$^{2}$Tree containing $u$ (Lines 22--23). 

\begin{lemma}
    The time complexity of \refalg{2eddeletent} is $O(h^{2}d)$, where $d$ is average degree of a vertex.
\end{lemma}

\begin{proof}
    The number of edges in tree path between $u$ and $v$ is $O(h)$. $Split2ECC$ can be called $O(h)$ times. In $Split2ECC$, we need to search the subtree rooted at $u$, which has $h$ vertices. $Split2ECC$ can be done in $O(hd)$. The time complexity of \refalg{2eddeletent} is $O(h^{2}d)$.
\end{proof}

\begin{algorithm}[t!]
\caption{\algtddeletet}
\label{alg:2eddeletet}
% \SetAlgoVlined
\KwIn{an existing edge $(u,v)$ and the \opttidxname index}
\KwOut{the updated \opttidxname}

Lines 1--6 of \refalg{2edeletet}\;
$P \gets Getrep(ru,rv)$\;
\If{$\left | P\right |>1$}{
$\algtdeletet(x,y)$\;
}
\Else{

\ForEach{$(x,y) \in P$}{

$\algtddeletent(x,y)$\;

        }
$\algcut(u,v)$\;

\ForEach{$(x,y) \in P$}{

$\algtedinsert(x,y)$\;

        }
        }

\end{algorithm}

\stitle{Tree edge deletion.} We divide tree edge deletions into three cases: 

\begin{itemize}
    \item 1. Deleting a tree edge whose $rep$ is greater than 1.
    
    \item 2. Deleting a tree edge whose $rep$ equals 0.

    \item 3. Deleting a tree edge whose $rep$ equals 1.
\end{itemize}

\begin{lemma}
\label{lem:del2}
    Given a tree edge $(u,u.\parent)$, if $rep$ of  $(u,u.\parent)$ is greater than 1, the 2-edge connected component containing $u$ and $u.\parent$ does not change.
\end{lemma}

\begin{proof}
    We let the current 2-edge-connected component containing $u$ and $u.\parent$ be $ECC$. Deleting the tree edge $(u, u.\parent)$ does not affect the 2-edge connectivity between $u$ and its descendants that are 2-edge-connected to $u$. Similarly, the 2-edge connectivity among the other vertices in $ECC$, including $u.parent$, also remains unaffected. Since there is more than one replacement edge for $(u, u.\parent)$, $u$ and $u.parent$ remain 2-edge-connected after the deletion. Therefore, the 2-edge connected component containing $u$ and $u.\parent$ does not change.

\end{proof}

Based on \reflem{del2}, the 2-edge connected component would not be changed in the first case. We can just update the \basetidxname. We do not require that a 2-edge connected component has the same root in the \basetidxname and the DS$^{2}$Tree. Therefore, there is no need to update the DS$^{2}$Tree. 

For the second case of tree edge deletion, $u$ and $v$ are not in the same 2-edge connected component. We just cut $(u,v)$ without updating DS$^{2}$Tree.

For the third case of tree edge deletion, $u$ and $v$ will no longer be in the same 2-edge connected component. Similar to \refalg{2einsert}, we first find the replacement edge and delete it with \refalg{2eddeletent}. Next, we can cut $(u,v)$ as a tree edge without replacement edges. Finally, we can reinsert the deleted replacement edge with \refalg{2edinsert}.

We propose our edge deletion algorithm of \opttidxname in \refalg{2eddeletet}. We compute all replacement edges of $(u,v)$ and store them in $P$ (Lines 1--2). For the first case, i.e., when $|P|>1$, we can only update the \basetidxname with \refalg{2edeletet} (Lines 3--4). For the other two case, i.e., when $|P|=1$ or 0, we delete the edges in $P$ with \refalg{2eddeletent} (Lines 6--7) and cut $(u,v)$ (Line 8). Finally, we reinsert the deleted replacement edge with \refalg{2edinsert}.

\begin{lemma}
    The time complexity of \refalg{2eddeletet} is $O(h^{2}d)$.
\end{lemma}

\begin{proof}
    For the first case in tree edge deletion, We only update the \basetidxname. Therefore, the complexity is the same as \refalg{2edeletet}. It is $O(h^{2}d)$. For the second case, it can be done in $O(1)$. For the third case, there only one edge in $P$. We need to delete this edge with \refalg{2eddeletent} in $O(h^{2}d)$ and cut $(u,v)$ in $O(1)$. Finally, we reinsert this edge with \refalg{2edinsert} in $O(h)$. The total time complexity of the third case is still $O(h^{2}d)$. 
\end{proof}

\stitle{The reason for not keeping the roots consistent.} For the first case of tree edge deletion, we only update ID$^{2}$Tree but not DS$^{2}$Tree. However, if we still want to keep the roots consistent, we also need to update DS$^{2}$Tree. Since the structure of the spanning tree has changed, the roots of many ID$^{2}$Trees are also changed. The roots of many DS$^{2}$Tree, including the one containing $u$ and $v$, must be updated accordingly. At this point, we need to traverse the entire spanning tree to find all ID$^{2}$Trees whose roots have changed. This operation is time-consuming and increases the overall time complexity of the algorithm. In addition, the two benefits of maintaining consistent roots discussed in Section 6.4 do not apply to updates in \opttidxname. On one hand, efficiently locating the root of a ID$^{2}$Tree alone is insufficient in \opttidxname, as many update operations require traversing the edges of a tree path. On the other hand, as described in the introduction of $Split2ECC$, we provide a method to avoid deleting root without increasing complexity. Therefore, maintaining consistent roots offers no benefit for \opttidxname.

% ID2-tree
% basic operator
% insert
% deletenon
% deletetree(rep search)
% DND2-tree
% basic operator
% insertnew
% deletenon(subtree search)

% deletetree

%!TEX root = main.tex
\section{Performance Studies}
\label{sec:exp}

	\begin{table*}[t!]
		\centering
		\begin{small}
  \resizebox{\linewidth}{!}{
			
			\begin{tabular}{c|c|c|c|c|c|c|c|c|c}
				\hline
				Dataset           & name & $n$         & $m$           & Type                    & $h_{\basetree}$ & $h_{\dtree}$ &$|S|$& $\#search$ & \rr{$h_{ID^{2}Tree}$}  \\ \hline\hline
fb-forum\footnotemark[3]    & FO   & 900   & 33,720    
& Temporal, Social      & 2.834     & 2.776 &1.112 &1.038 &1.629\\ \hline
ia-reality-call\footnotemark[3]    & IA   & 27,045    & 52,050    
& Temporal, Communications       & 1.578     & 1.001 &1.226 &1.578 &1.188\\ \hline
dynamic-dewiki\footnotemark[1]    & DE   & 2,166,670   & 86,337,879    
& Temporal, Hyperlink     & 2.828     & 2.778 &1.015 &1.001 &3.856\\ \hline
stackoverflow\footnotemark[2]     & ST   & 2,601,978     & 63,497,050       
&Temporal,~
QA          & 2.509     & 2.445 &1.037&1.005 &19.932 \\ \hline
% soc-bitcoin\footnotemark[3]       & BI   & 24,575,383  & 122,948,162    
% &Temporal,~
% Transaction & 5.384     & 6.350  &1.133&1.477 &389.598\\ \hline
soc-flickr-growth\footnotemark[3] & FL   & 2,302,926   & 33,140,017     
&Temporal,~
Social      & 1.163     & 1.190 &1.424&1.055 &4.340\\ \hline
edit-enwiki\footnotemark[1]       & EN   & 50,757,444  & 572,591,272    
&Temporal,~
Edit        & 2.472     & 2.160 &1.008&1.002 &8.528\\ \hline
delicious-ti\footnotemark[1]      & TI   & 38,289,742  & 301,183,605    
&Temporal,~
Feature     & 2.600       & 2.592 &1.051&1.029 &8.192\\ \hline
delicious-ui\footnotemark[1]      & UI   & 34,611,304  & 301,186,579    
&Temporal,~
Interaction & 3.550      & 3.321  &1.012&1.018 & 10.279\\ \hline
yahoo-song\footnotemark[1]        & YA   & 1,625,953   & 256,804,235    
&Temporal,~
Rating      & 2.197     & 2.554 &1.000&1.000 &2.261\\ \hline
Slashdot0811\footnotemark[2]       & S8   & 77,360   & 469,180     
& Unlabeled, Social          & 2.672     & 4.064 &1.059&1.002
 & 2.728\\ \hline
 Gowalla\footnotemark[2]       & GO   & 196,592   & 950,327     
& Unlabeled, Location          &  2.571     & 4.112 &1.173&1.022
 & 2.946\\ \hline
LiveJournal\footnotemark[1]       & LI   & 4,846,610   & 42,851,237     
& Unlabeled, Social          & 3.911     & 4.272 &1.068&1.002
 &3.954\\ \hline
twitter\_mpi\footnotemark[1]      & TM   & 52,579,683  & 1,614,106,187  
&Unlabeled,~
Social        & 2.353     & 2.462 &1.015&1.000
 &2.353\\ \hline
twitter-2010\footnotemark[1]      & T2   & 41,652,230  & 1,202,513,046  
&Unlabeled,~
Social        & 2.334     & 3.470 &1.013&1.000
&2.334 \\ \hline
friendster\footnotemark[1]        & FR   & 124,836,180 & 1,806,067,135  
&Unlabeled,~
Social        & 1.855     & 2.640 &1.047&1.003
&1.855\\ \hline
uk-2007\footnotemark[4]           & UK   & 133,633,040 & 4,663,392,591  
&Unlabeled,~
Hyperlink     & 4.003     & 5.249 &1.818&1.002
 &4.004\\ \hline
CTR\footnotemark[5]               & CT   & 14,081,817  & 16,933,413     
&Unlabeled,~
Road          & 2,356.592  & 1,883.230 &3.279&6.541
 &2588.938\\ \hline
W\footnotemark[5]                 & W    & 6,262,105   & 7,559,642      
&Unlabeled,~
Road          & 1,410.453  & 1,404.410 &4.194&10.840
 &1987.076\\ \hline
road-usa\footnotemark[3]     & US   & 23,947,348  & 28,854,312     
&Unlabeled,~
Road          & 2,864.871  & 2,726.730 &3.202&6.791
 &2865.402\\ \hline
% sc-delaunay\_n23\footnotemark[3]  & N2   & 8,388,609   & 25,165,784    & 3.000 
% &Unlabeled,~
% Scientific    & 356.125   & 431.030  &0.000&1.000 \\ \hline
% sc-ldoor\footnotemark[3]          & LD   & 952,204     & 20,770,807    & 21.813 
% &Unlabeled,~
% Scientific    & 77.718    & 94.290 &0.000&1.000 \\ \hline
% sc-rel9\footnotemark[3]           & RE   & 5,921,787   & 23,667,162    & 3.997 
% &Unlabeled,~
% Scientific    & 5.434     & 5.675 &0.000&1.001      \\\hline
			\end{tabular}
   }
		\end{small}
  \caption{The Description of Dataset. $h$ is the average vertex depth in different spanning trees. $m$ is the number of edges and $n$ is the number of vertices. $|S|$ is the average size of S in \refalg{optdelete}. $\#search$ is the actual average number of iterations which is mentioned in \reflem{expect}. 
  %\rr{$maxd$ is the maximum depth of vertex on DS-tree during query phase.}
  } 
			\label{tab:dataset}
			\vspace{-1em}
	\end{table*}

% \textcolor{red}{DND-trees}

% \textcolor{red}{delete table 2 last 1 column}

% \textcolor{red}{Figure 7, Figure 8 reduce height}

% \textcolor{red}{sliding window time span? how to calculate...}

% \textcolor{red}{figure y-label add more}

% \textcolor{red}{add space before (}

% \textcolor{red}{scalability FR not very good}

% \textcolor{red}{Family name et al.}

% \textcolor{red}{grammar error}

\stitle{Setup.} \rr{All algorithms are implemented in C++ and compiled with O3 level optimization. The experiments are conducted on a single machine with Intel Xeon Gold 6248 2.5GHz and 768GB RAM. All results are averaged over ten runs on the same machine.} 
%\textcolor{red}{\sout{If the program cannot be completed within 12h, we record as INF.}}

\stitle{Dataset.} We evaluate many real datasets from different domains (Table \ref{tab:dataset}). These datasets can be found at konect\footnotemark[1], Stanford Large Network data set Collection\footnotemark[2], Network Repository\footnotemark[3], LAW\footnotemark[4] and DIMACS\footnotemark[5]. Nine out of nineteen datasets are temporal graphs, and the rest are unlabeled graphs. \rr{Our algorithms can be used for recommendation in social networks \cite{mirza2003studying}, transaction analysis in trading networks \cite{maesa2016uncovering}, path planning in road networks \cite{daniel2020gis}, etc.}

\stitle{Competitors.}
\rr{We evaluate the performance of queries and update operations for the following methods:}
\begin{itemize}
% \vspace{-0.4em}
	\item \textbf{\optidxname.} Our final algorithm includes all optimizations.
	\item \textbf{\baseidxname.} Our algorithm shown in section \ref{sec:base}.
	\item \textbf{\dtree.} The algorithm proposed by  Chen et al. \cite{DBLP:journals/pvldb/ChenLHB22}. The source code from \cite{DBLP:journals/pvldb/ChenLHB22} is in Python. We re-implement it in C++.
    \item \textbf{\idtdsa.} A baseline to combine \baseidxname and the disjoint set. When a tree edge is deleted in \baseidxname and no replacement edge exists, we reconstruct the whole disjoint set.
    % As mentioned in section 4, This basic method that combine the \baseidxname index and disjoint-set structures by traversing all trees of the \baseidxname index when a replacement cannot be found.
    \item \textbf{\idtdsb.} The other baseline to combine \baseidxname and the disjoint set. When a tree edge is deleted in \baseidxname and no replacement edge exists, we reconstruct the disjoint set for all vertices in the old connected component containing the deleted edge.

	%  The authors' code of \cite{DBLP:journals/pvldb/ChenLHB22} is implemented in python. For fairness, we reimplemented it in C++. 
	% After testing, the efficiency of the algorithm after rewriting it in C++ is  higher than that of python code. The experimental results of D-Tree in this article show the results of C++ code.
    \item \textbf{\hdt.} The algorithm proposed by Holm et al. \cite{DBLP:conf/stoc/HolmLT98,DBLP:journals/jacm/HolmLT01}. The code is from an experimental paper \cite{iyer2001experimental}.
    
    \item \textbf{DND$^{2}$Tree.} Our final algorithm for 2-edge connectivity includes all optimizations (section 8).
    
    \item \textbf{ID$^{2}$Tree.} Our algorithm shown in (section 7).

    \item \textbf{HDT$^{2}$.} The algorithm for 2-edge connectivity proposed by Holm et al. \cite{DBLP:conf/stoc/HolmLT98,DBLP:journals/jacm/HolmLT01}. 

    \item \textbf{Rec$^{2}$.} A recomputation-based algorithm for maintaining 2-edge connectivity. It uses two groups of disjoint-Set Structures to maintain connected components and 2-edge-connected components. When a new edge is inserted, the corresponding sets are merged. When an old edge is deleted, all connected components and all 2-edge-connected components are recomputed.
\end{itemize}

	\footnotetext[1]{http://konect.cc/networks/} 
	\footnotetext[2]{https://snap.stanford.edu/data/} 
	\footnotetext[3]{https://networkrepository.com/} 
	\footnotetext[4]{http://law.di.unimi.it/datasets.php} 
 \footnotetext[5]{http://www.diag.uniroma1.it/challenge9/download.shtml}

% \subsection{Algorithms}

 \begin{figure*}[t!]
	\centering
	\includegraphics[width=1\textwidth]{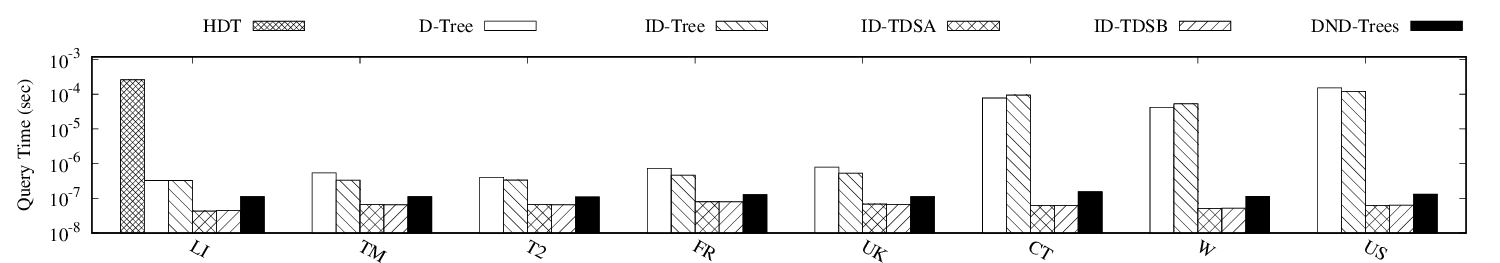}
	\vspace{-2em}
	\caption{Query time of unlabeled graphs.}
	\label{fig:Query time of unlabeled graphs}
	%\vspace{-0.5em}
\end{figure*}
	\begin{figure*}[t!]
	\centering
	\includegraphics[width=1\textwidth]{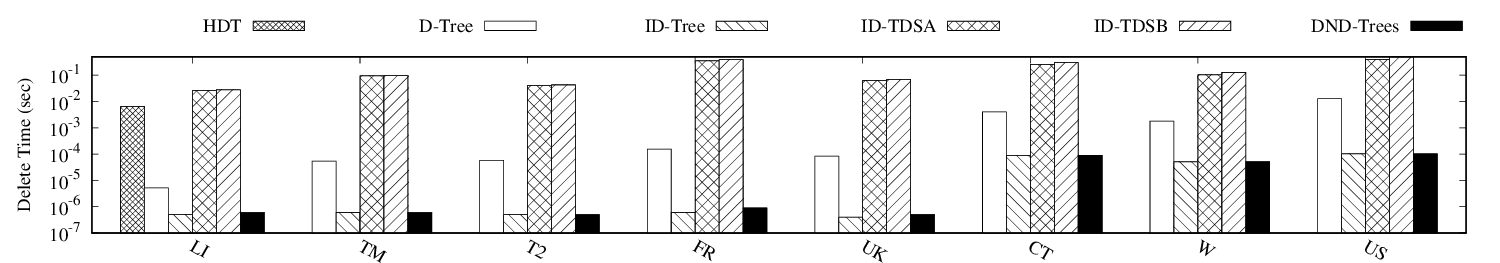}
	\vspace{-2em}
	\caption{Delete time of unlabeled graphs.}
	\label{fig:Delete time of unlabeled graphs}
	%\vspace{-0.5em}
\end{figure*}
\begin{figure*}[t!]
	\centering
	\includegraphics[width=1\textwidth]{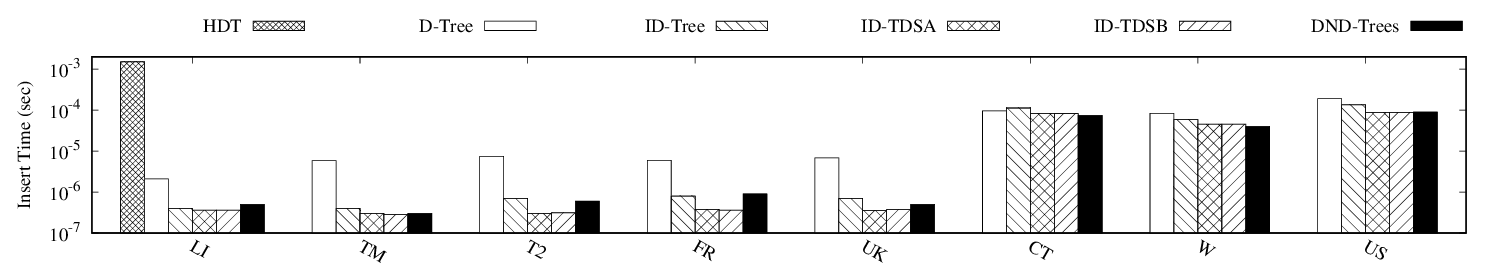}
	\vspace{-2em}
	\caption{Insert time of unlabeled graphs.}
	\label{fig:Insert time of unlabeled graphs}
	%\vspace{-0.5em}
\end{figure*}
\begin{figure*}[t!]
	\centering
	\includegraphics[width=1\textwidth]{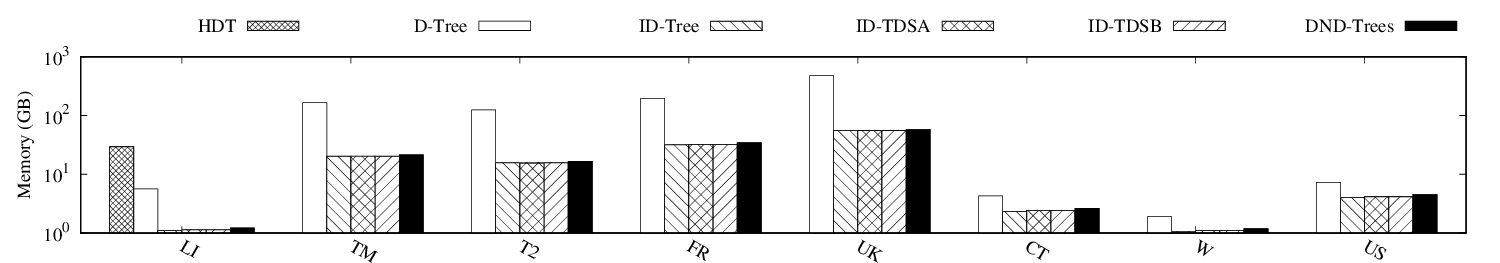}
	\vspace{-2em}
	\caption{Memory of unlabeled graphs.}
	\label{fig:Memory of unlabeled graphs}
	%\vspace{-0.5em}
\end{figure*}

\vspace{-0.7em}
\subsection{Performance in Unlabeled Graphs}

For a general unlabeled graph, we first build an index with all the edges of the entire dataset. Then we randomly delete $100,000$ edges and then insert those $100,000$ edges back into the graph. We calculate the average running time for insertions and deletions, respectively.
For query efficiency, we randomly generate $50,000,000$ vertex pairs and identify whether they are connected. 
Due to the large average vertex depth of the road network, queries are not efficient. For them we query $100,000$ times. We calculate the average time of a query operation. For any of the test cases in our paper, \optidxname can complete the entire operation in less than $1,000$ seconds. \rr{We do not report the results of tests that take longer than 12 hours. A case executes for more than 12h is mostly due to its very low deletion efficiency, indicating that the algorithm is not suitable for processing large-scale data.} All subsequent experiments also follow this setting.

\stitle{Query Processing.} 
Figure \ref{fig:Query time of unlabeled graphs} shows the average query time in 8 unlabeled graphs. The query efficiency of \optidxname is much higher than that of \baseidxname and \dtree, and very close to \idtdsa and \idtdsb. For the instance of US, \optidxname is two orders of magnitude faster. 
The efficiency of both \baseidxname and \dtree is related to the average depth. 
Their query time is different but the overall query efficiency of \baseidxname and \dtree is similar.

\stitle{Deletion.}
Figure \ref{fig:Delete time of unlabeled graphs} shows the average time of the edge deletion. On all datasets, the deletion efficiency of \optidxname is significantly higher than that of \dtree. Two orders of magnitude speedup is achieved on TM, UK and US. We can also see that the additional cost of \optidxname beyond \baseidxname is marginal, which supports \refthm{finalinsert} and \refthm{finaldelete}. \idtdsa and \idtdsb are much slower than \optidxname, because they need to visit nearly all the tree edges in \baseidxname.
	
\stitle{Insertion.} Figure \ref{fig:Insert time of unlabeled graphs} shows the average time of edge insertion. Our final algorithm is much faster than \dtree. Note that in certain small datasets (LI and FR), \optidxname is a little slower than \baseidxname. This is because when the average depth of the \baseidxname itself is relatively small, although the \dstree speeds up the root-finding operation, additional time is required to maintain the \dstree.

\stitle{Memory.} Figure \ref{fig:Memory of unlabeled graphs} shows the memory usage. Compared to \dtree, \baseidxname does not need to maintain children and non-tree neighbors and its memory is smaller. Compared to \baseidxname, \optidxname include \dstree structure and need a little more memory.

\stitle{Different types of update operations.} \rr{Figure \ref{fig:Update time of tree and non-tree edges} shows the efficiency of different types of update operations (inserting tree edges, inserting non-tree edges, deleting tree edges and deleting non-tree edges). We first initialize the spanning tree based on all edges in the dataset and obtain the set of tree edges and the set of non-tree edges. Next, we delete all non-tree edges and calculate the efficiency of non-tree edge deletion. Then, all tree edges are deleted, and the efficiency of tree edge deletion is calculated. Next, we insert all tree edges and calculate the efficiency of tree edge insertion. Finally, we insert all non-tree edges and calculate the efficiency of non-tree edge insertion. We can find that the update efficiency of tree edges is faster than that of non-tree edges, and the gap increases when $h$ increases.}

\begin{figure}[t!]
	\centering
	\includegraphics[width=0.4\textwidth]{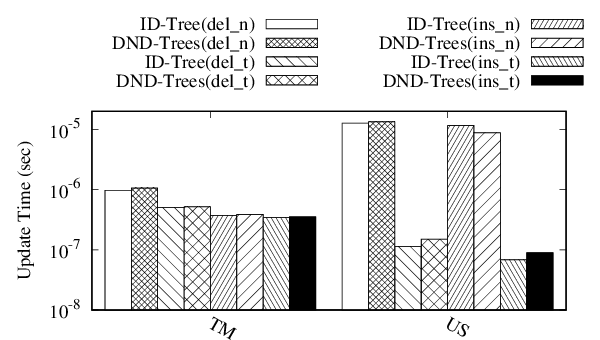}
	\vspace{-0.6em}
	\caption{\rr{Update time of tree and non-tree edges.}}
	\label{fig:Update time of tree and non-tree edges}
	\vspace{-0.5em}
\end{figure}
	
\begin{figure*}[t!]
	\centering
	\begin{tabular}{ccc}
		\subfigure[ST]{\includegraphics[width=0.3\textwidth]{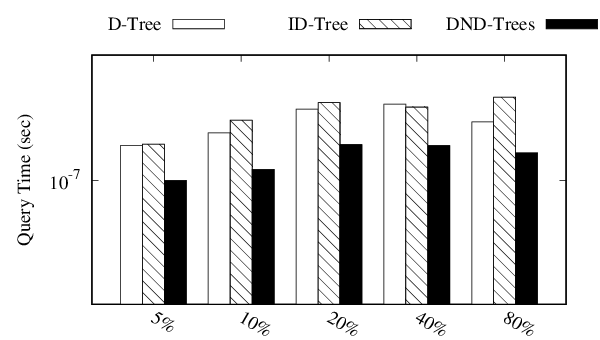}} & 
		\subfigure[UI]{\includegraphics[width=0.3\textwidth]{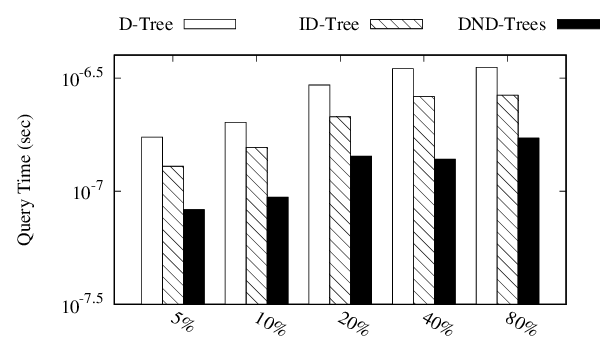}} &
		\subfigure[TI]{\includegraphics[width=0.3\textwidth]{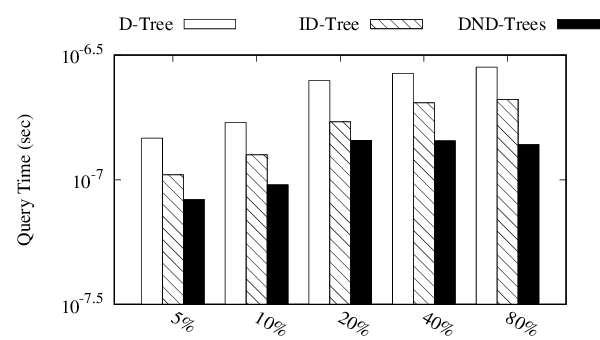}}
		
	\end{tabular}
 \vspace{-0.7em}
	\caption{Query time (vary window size).}
	\label{fig:Query time (vary window size)}
 %\vspace{-1.5em}
\end{figure*}

% \begin{figure*}[t!]
% 	\centering
% 	\begin{tabular}{ccc}
% 		\subfigure[\baseidxname]{\includegraphics[width=0.35\textwidth]{soc-bitcoinqnew.eps}} & 
% 		\subfigure[\dstree]{\includegraphics[width=0.35\textwidth]{soc-bitcoinqnew1.eps}} 
		
% 	\end{tabular}
%  \vspace{-0.7em}
% 	\caption{\rr{Query time in different windows.}}
% 	\label{fig:Query time in different windows}
%  %\vspace{-1.5em}
% \end{figure*}

% \begin{figure}[t!]
% 	\centering
% 	\includegraphics[width=1\textwidth]{soc-bitcoinqnew.eps}
% 	\vspace{-0.6em}
% 	\caption{\rr{Query time by varying window size (\baseidxname).}}
% 	\label{fig:Query time in different windows baseidxname}
% 	\vspace{-0.5em}
% \end{figure}

% \begin{figure}[t!]
% 	\centering
% 	\includegraphics[width=1\textwidth]{soc-bitcoinqnew1.eps}
% 	\vspace{-0.6em}
% 	\caption{\rr{Query time by varying window size (\dstree). }}
% 	\label{fig:Query time in different windows dstree}
% 	\vspace{-0.5em}
% \end{figure}

\begin{figure*}[t!]
	\centering
	\begin{tabular}{ccc}
		\subfigure[ST]{\includegraphics[width=0.3\textwidth]{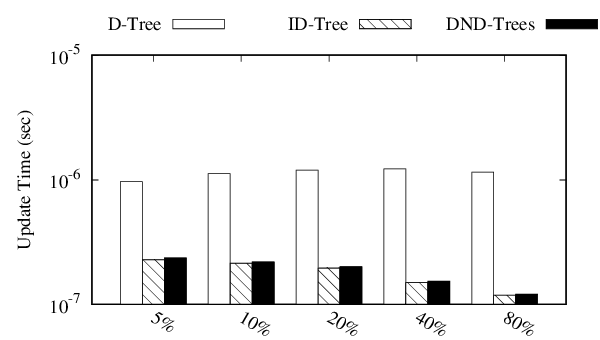}} & 
		\subfigure[UI]{\includegraphics[width=0.3\textwidth]{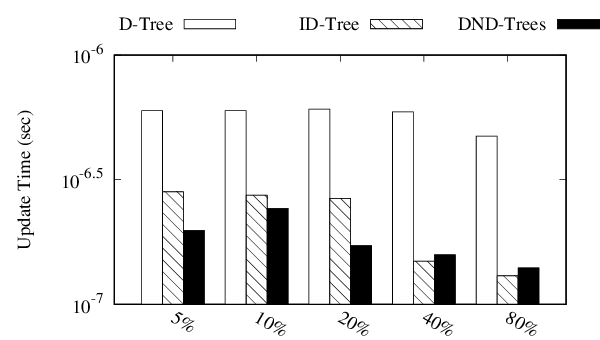}} &
		\subfigure[TI]{\includegraphics[width=0.3\textwidth]{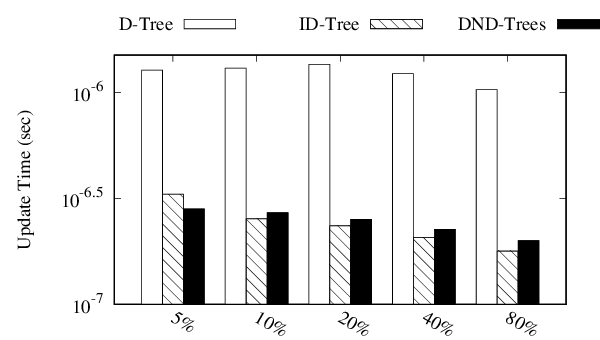}} 
		
	\end{tabular}
 \vspace{-0.7em}
	\caption{Update time (vary window size).}
	\label{fig:Update time (vary window size)}
 \vspace{-1em}
\end{figure*}

\vspace{-1em}
\subsection{Performance in Sliding Windows}
\label{subsec:pps}

We investigate three algorithms over different sizes of sliding windows in temporal graphs. For each temporal graph, we first compute the time span for the dataset. Then we vary the window size in 5\%, 10\%, 20\%, 40\% and 80\% of its time span. We insert all edges in chronological order. When the time difference between the inserted new edge and the oldest edge in the window is greater than the time window size, the old edge is deleted.  We record the average time of a sliding operation (inserting a new edge and deleting an expired edge). We also randomly query 50,000,000 times in the last time window and calculate the average query time.
We report the representative datasets given the space limitation.

% We record the time from inserting the first edge to the sliding window to the sliding of the window to the end of the data set and the number of all insertion and deletion operations to calculate the average time of update operations (including insertion and deletion).

\stitle{Query processing by varying window size.}
Figure \ref{fig:Query time (vary window size)} shows the average query time of different window sizes. As the time window size increases, the query time of both \dtree and \baseidxname increases. However, the query time of \optidxname is stable given its constant query time. 
%\rr{Figure \ref{fig:Query time in different windows} shows the average query time in five different sliding windows. Due to space constraint, we chose dataset BI as an example. We select 5 windows evenly for querying during the sliding process, depending on the window size. Specifically, when the window size is 40\%, we select the following five windows: 0\%--40\%, 15\%--55\%, 30\%--70\%, 45\%--85\% and 60\%--100\%. Some of the earlier windows have fewer temporal edges, and their queries are more efficient, as influenced by the distribution of temporal edges. However, their query efficiency stabilizes when there are enough temporal edges in the window.}
	
\stitle{Updating by varying window size.}
Figure \ref{fig:Update time (vary window size)} shows the average update time of different window sizes. The \optidxname and \baseidxname are much faster than \dtree. As the window size increases, the time variation of all three algorithms is not obvious. \rr{For the sliding window updates, deleting an edge is the dominating cost because of searching the replacement edge. As the window size increases, the number of edges increases, and the graph becomes denser. As a result, the probability of deleting a non-tree edge is higher. Compared with the case of deleting a tree edge, deleting a non-tree edge is much more efficient.}
%\rr{We present experiments on all temporal graphs in the technical report \cite{report}.}

\subsection{Average depth of Spanning Tree}

The efficiency of all algorithms is related to the average tree depth $h$ except for our final query algorithm. 
$h_{\baseidxname}$ and $h_{\dtree}$ in Table \ref{tab:dataset} show the average tree depths in different graphs. 
%$maxd$ in Table \ref{tab:dataset} shows the maximum depth of vertex on \dstree during query phase. 
For temporal datasets, they refer to the average depth where the time window is set to 40\%. Even though we relax certain heuristics to reduce average tree depth, the result shows that our tree depth is still competitive. Furthermore, experimental results show that the depth of the \dstree can be considered as a constant.

% The main difference between ID-Tree, DS-Trees and D-Tree is that the optimization of non-tree edge insertion and the deletion operation stop immediately when the replacement edge is found. However, the experimental results show that on the eight data sets, the difference between the average depth of DS-Trees and D-Tree is not large. This shows that our optimization not only reduces the update time, but also ensures that the average depth of the spanning tree is basically unchanged.

% Similar to the results for unlabeled graphs, the average depths of DS-Trees and D-Tree in temporal graphs are also not much different. 

\vspace{-0.7em}
\subsection{Experimental Analysis of Delete Operation}
	
\reftab{dataset} reports the average value of $|S|$ in \refalg{optdelete} when $succ$ is $true$ in Line 3. For temporal datasets, $h$ refers to the case where the time window is set to 40\% of the time span. 
Compared to \baseidxname, the additional cost of \optidxname appears when no replacement edge can be found. We disconnect the \dstree as shown in \refalg{optdelete}. The cost depends on the size of $S$. From \reftab{dataset}, $|S|$ is very small which supports our theoretical results and proves that the additional cost is negligible.  We also report the actual average number of iterations of Line 9 in \refalg{basedelete}. As shown in Table \ref{tab:dataset}, $\#search$ is small and it is in line with Lemma \ref{lem:expect}. This demonstrates the high efficiency of our edge deletion operation.

 \begin{figure*}[t!]
	\centering
	\includegraphics[width=1\textwidth]{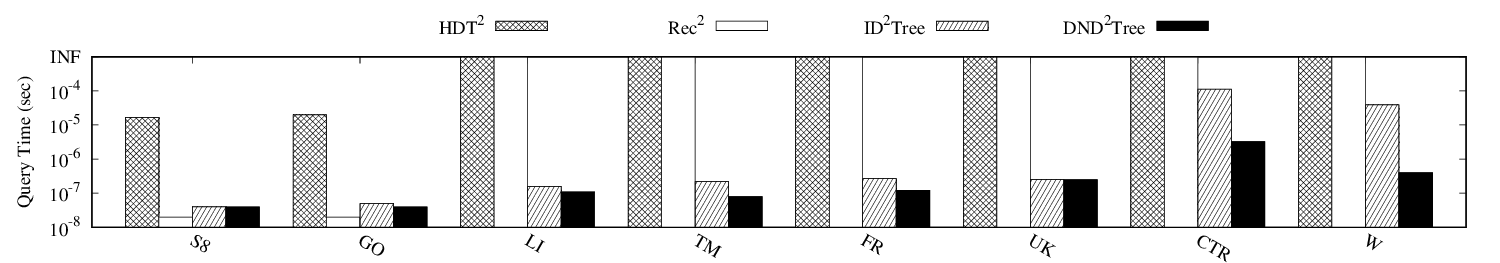}
	\vspace{-2em}
	\caption{Query time of unlabeled graphs (2-edge connectivity).}
	\label{fig:Query time of unlabeled graphs (2-edge connectivity)}
	\vspace{-0.1em}
\end{figure*}

 \begin{figure*}[t!]
	\centering
	\includegraphics[width=1\textwidth]{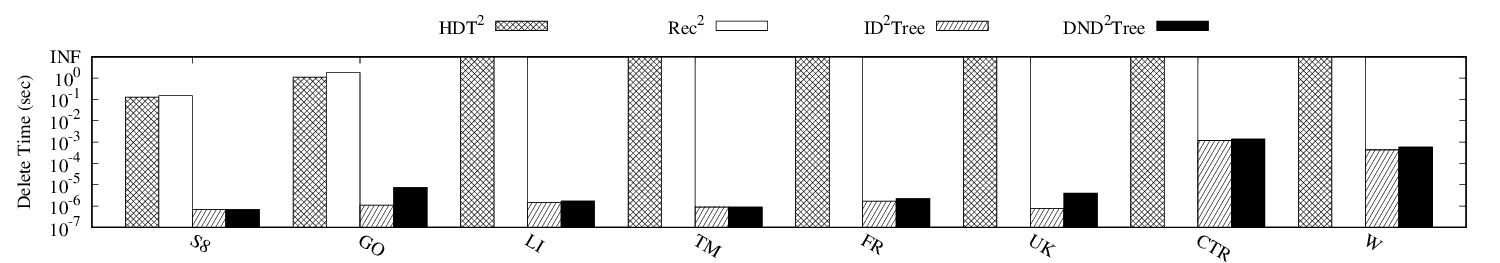}
	\vspace{-2em}
	\caption{Delete time of unlabeled graphs (2-edge connectivity).}
	\label{fig:Delete time of unlabeled graphs edge connectivity}
	\vspace{-0.1em}
\end{figure*}

 \begin{figure*}[t!]
	\centering
	\includegraphics[width=1\textwidth]{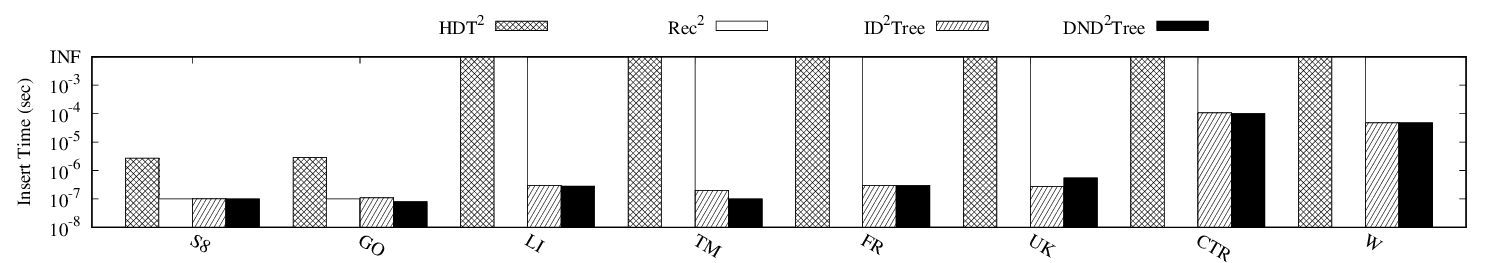}
	\vspace{-2em}
	\caption{Insert time of unlabeled graphs (2-edge connectivity).}
	\label{fig:Insert time of unlabeled graphs (2-edge connectivity)}
	\vspace{-0.1em}
\end{figure*}

\begin{figure*}[t!]
	\centering
	\includegraphics[width=1\textwidth]{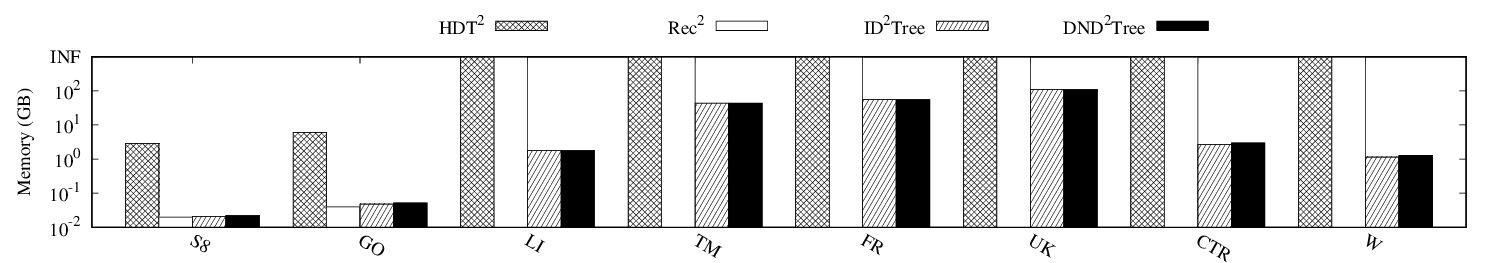}
	\vspace{-2em}
	\caption{Memory of unlabeled graphs (2-edge connectivity).}
	\label{fig:Memory of unlabeled graphs (2-edge connectivity)}
	\vspace{-0.1em}
\end{figure*}

 \begin{figure*}[t!]
	\centering
	\includegraphics[width=1\textwidth]{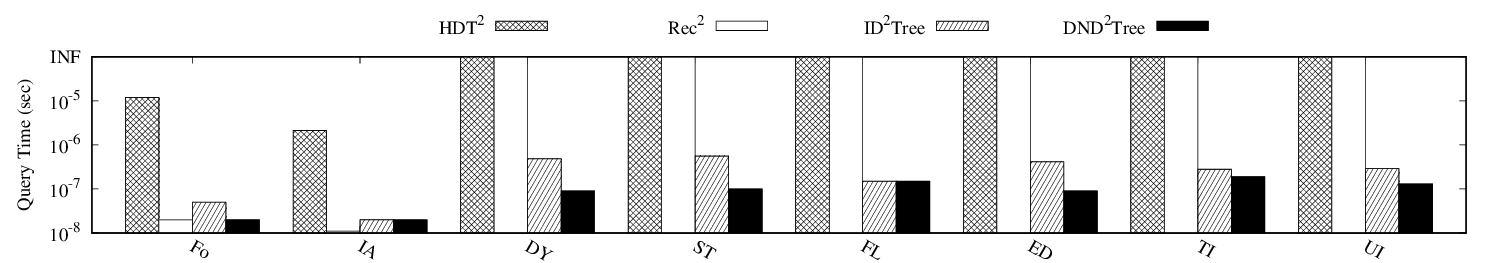}
	\vspace{-2em}
	\caption{Query time of temporal graphs (2-edge connectivity).}
	\label{fig:Query time of temporal graphs (2-edge connectivity)}
	%\vspace{-0.5em}
\end{figure*}

 \begin{figure*}[t!]
	\centering
	\includegraphics[width=1\textwidth]{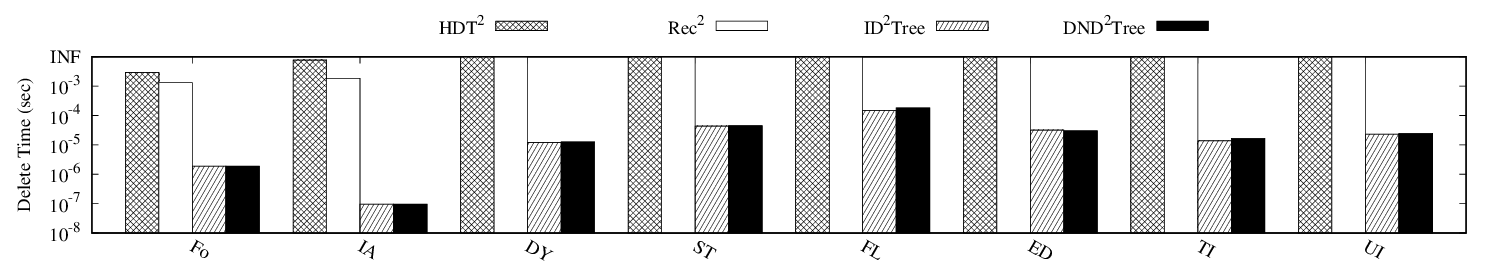}
	\vspace{-2em}
	\caption{Delete time of temporal graphs (2-edge connectivity).}
	\label{fig:Delete time of temporal graphs (2-edge connectivity)}
	\vspace{-1em}
\end{figure*}

 \begin{figure*}[t!]
	\centering
	\includegraphics[width=1\textwidth]{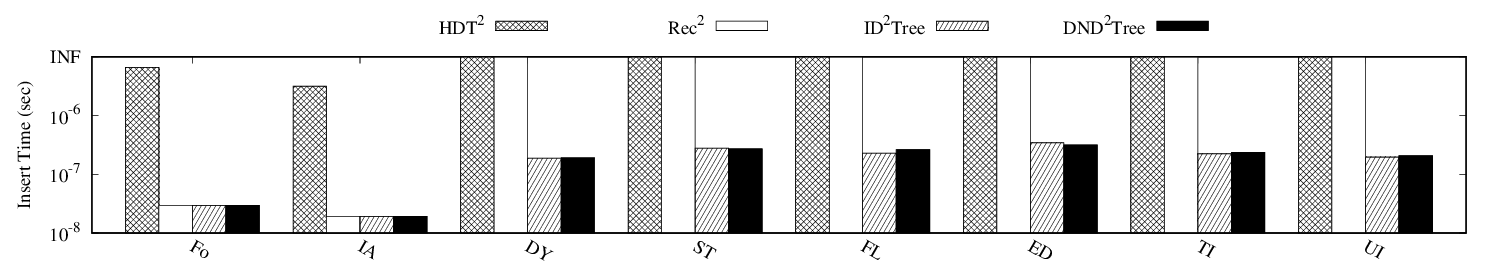}
	\vspace{-2em}
	\caption{Insert time of temporal graphs (2-edge connectivity).}
	\label{fig:Insert time of temporal graphs (2-edge connectivity)}
	%\vspace{-0.5em}
\end{figure*}

\begin{figure*}[t!]
	\centering
	\includegraphics[width=1\textwidth]{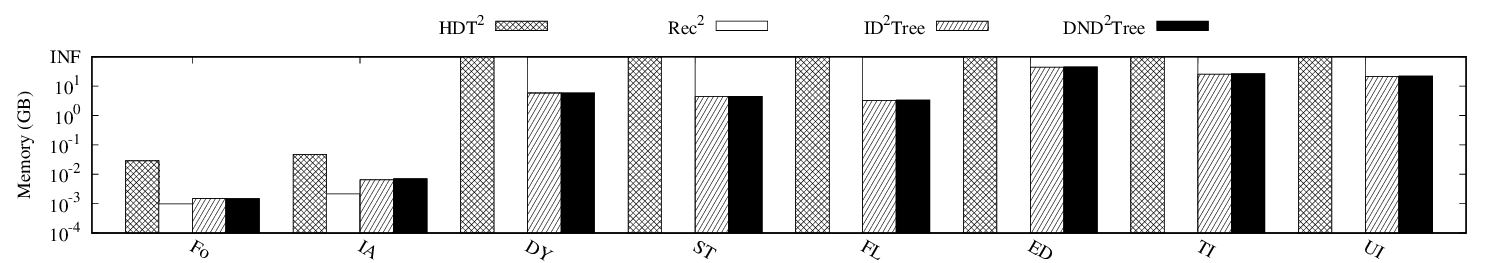}
	\vspace{-2em}
	\caption{Memory of temporal graphs (2-edge connectivity).}
	\label{fig:Memory of temporal graphs (2-edge connectivity)}
	%\vspace{-0.5em}
\end{figure*}

\begin{figure*}[t!]
	\centering
	\begin{tabular}{ccc}
		\subfigure[FO]{\includegraphics[width=0.3\textwidth]{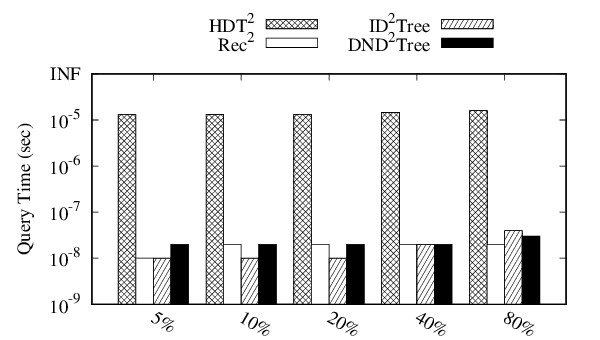}} & 
		\subfigure[IA]{\includegraphics[width=0.3\textwidth]{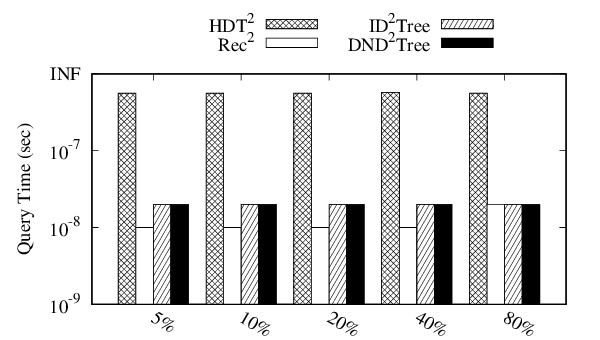}} &
		\subfigure[ST]{\includegraphics[width=0.3\textwidth]{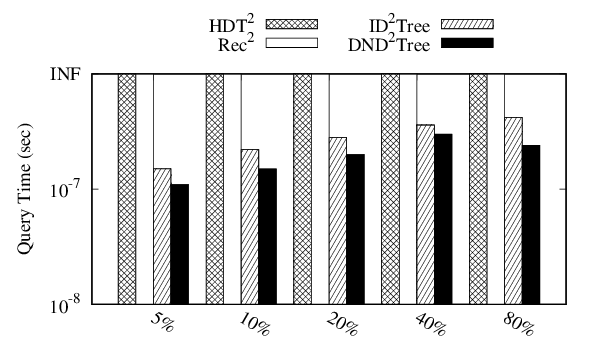}}
		
	\end{tabular}
 \vspace{-0.7em}
	\caption{Query time for 2-edge connectivity (vary window size).}
	\label{fig:Query time for 2-edge connectivity (vary window size)}
 %\vspace{-1.5em}
\end{figure*}

\begin{figure*}[t!]
	\centering
	\begin{tabular}{ccc}
		\subfigure[FO]{\includegraphics[width=0.3\textwidth]{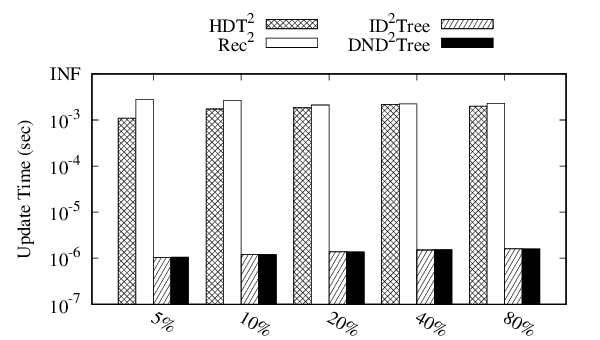}} & 
		\subfigure[IA]{\includegraphics[width=0.3\textwidth]{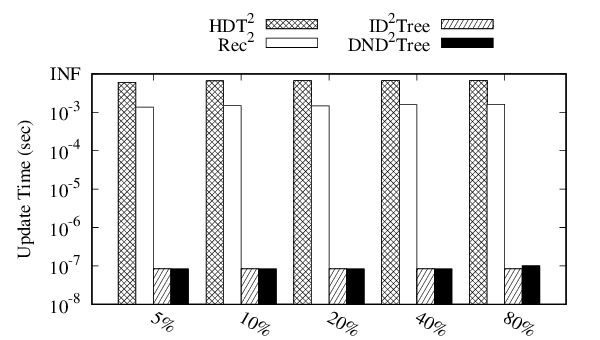}} &
		\subfigure[ST]{\includegraphics[width=0.3\textwidth]{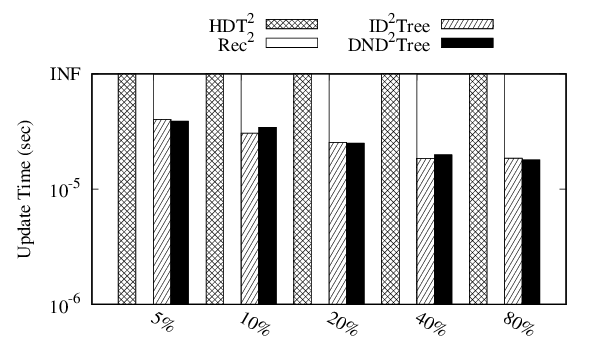}}
		
	\end{tabular}
 \vspace{-0.7em}
	\caption{Update time for 2-edge connectivity (vary window size).}
	\label{fig:Update time for 2-edge connectivity (vary window size)}
 \vspace{-1.5em}
\end{figure*}

\vspace{-0.7em}
\subsection{Performance for 2-edge connectivity}

\stitle{\textit{Unlabeled Graphs.}} For a general unlabeled graph, we first build an index with all the edges of the entire dataset. Then we randomly delete $10,000$ edges and then insert those $10,000$ edges back into the graph. We calculate the average running time for insertions and deletions, respectively.
For query efficiency, we randomly generate $100,000$ vertex pairs and identify whether they are 2-edge connected.

\stitle{Query Processing.} 
Figure \ref{fig:Query time of unlabeled graphs (2-edge connectivity)} shows the average query time in 8 unlabeled graphs. The query efficiency of \opttidxname is much higher than that of HDT$^{2}$, and very close to \basetidxname in most cases. According to \reftab{dataset}, $h_{ID^{2}Tree}$ in many datasets is not big. Therefore, the query efficiency of \basetidxname is not bad. In addition, although we use $h$ in the complexity analysis, for a query vertex we only need to examine the height of the \basetidxname containing that vertex, which is smaller than $h$ in practice. This makes the advantage of \opttidxname much less significant. However, $h$ is very big in some sparse graph. The advantage of \opttidxname becomes much more evident on sparse graphs.

\stitle{Deletion.}
Figure \ref{fig:Delete time of unlabeled graphs edge connectivity} shows the average time of the edge deletion. On all datasets, the deletion efficiency of \basetidxname and \opttidxname are significantly higher than others methods. The efficiency of \opttidxname is very close to \basetidxname. This indicates that the overhead of maintaining the DS$^{2}$Tree is very small.

\stitle{Insertion.} Figure \ref{fig:Insert time of unlabeled graphs (2-edge connectivity)} shows the average time of edge insertion. Our two algorithms are much faster than HDT$^{2}$. The efficiency of \opttidxname is also close to \basetidxname.

\stitle{Memory.} Figure \ref{fig:Memory of unlabeled graphs (2-edge connectivity)} shows the memory usage. Compared to \basetidxname, \opttidxname include \dstree structure and need a little more memory.

\stitle{\textit{Temporal Graphs.}} We also evaluate eight temporal graphs. We first insert all the edges one by one in the original chronological order and compute the average time of an insert operation. Then we randomly generate $100,000$ vertex pairs for query processing. Finally, we continuously
delete the oldest edge until the graph is empty. 
The average times of edge deletion, edge insertion and queries are reported in Figure \ref{fig:Query time of temporal graphs (2-edge connectivity)}, \ref{fig:Delete time of temporal graphs (2-edge connectivity)} and \ref{fig:Insert time of temporal graphs (2-edge connectivity)}, respectively. Figure \ref{fig:Memory of temporal graphs (2-edge connectivity)} shows the memory usage. The results of queries are similar to unlabeled graphs.

\stitle{\textit{Performance in Sliding Windows.}} We investigate four algorithms over different sizes of sliding windows in temporal graphs using the same methodology as in \refsubsec{pps}.

\stitle{Query processing by varying window size.}
Figure \ref{fig:Query time for 2-edge connectivity (vary window size)} shows the average query time of different window sizes. As the time window size increases, the query time of both \basetidxname and \opttidxname is stable. 

\stitle{Update processing by varying window size.}
Figure \ref{fig:Update time for 2-edge connectivity (vary window size)} shows the average update time of different window sizes. \basetidxname and \opttidxname are still faster than the other two methods in all cases.

% \stitle{Average depth of Spanning Tree} \reftab{dataset} shows the average depth of \basetidxname. Compared with \baseidxname, the depth of \basetidxname is not much larger.

%!TEX root = main.tex
\section{Related Work}
\label{sec:rel}

\stitle{Connectivity in Undirected Graphs.} 
Initially, connectivity algorithms were designed for updating spanning trees either for edge insertions \cite{tarjan1975efficiency} or for edge deletions \cite{shiloach1981line}. 
% The first algorithms for updating minimum spanning trees in fully dynamic, undirected, and weighted graphs were developed by Spira and Pan \cite{spira1975finding}, Chin and Houck \cite{chin1978algorithms}, and Frederickson \cite{frederickson1983data}. 
% The algorithm developed by Spira and Pan has a time complexity of $O(n)$ for insertions and $O(n^{3})$ for deletions, where $n$ is the number of vertices. Chin and Houck improved the complexity for deletions to $O(n^{2})$. Frederickson further reduced the complexity for both insertions and deletions to $O(\sqrt{m})$, where $m$ is the number of edges. Eppstein et al improved the complexity even further to $O(\sqrt{n})$ per update operation \cite{eppstein1997sparsification}, but they did not provide an implementation.
%
Henzinger and King proposed \cite{henzinger1995randomized,henzinger1999randomized} a method of representing spanning trees through Euler tours \cite{tarjan1984finding}. It adds information to enable early termination of the search for a replacement edge. Holm et al. \cite{DBLP:conf/stoc/HolmLT98,DBLP:journals/jacm/HolmLT01} proposed a new structure that makes the update efficiency of the index reach $O(\log^{2}{n})$. Huang et al. \cite{huang2023fully} further theoretically reduced the time complexity to $O(\log{n}(\log{\log{n}})^{2})$. Chen et al. \cite{DBLP:journals/pvldb/ChenLHB22} introduce a new data structure, called the dynamic tree (D-Tree). A detailed analysis of D-Tree can be found in Section \ref{subsec:dtree}. Connectivity maintenance in streaming graphs is studied in \cite{DBLP:journals/pacmmod/SongWXQZ024} and \cite{xu2026querying}.

\stitle{Connectivity/reachability in Directed Graphs.}
Much of the prior research on reachability of directed graphs \cite{bramandia2009incremental,cheng2013tf,jin20093,wei2014reachability,zhu2014reachability} focuses on labeling schemes. Those approaches are typically not suited for undirected graphs. They can be classified into two main categories: interval labeling and 2-HOP labeling \cite{lyu2021dbl}. Those methods can also be modified for dynamic graphs. Optimal Tree Cover (Opt-TC) \cite{agrawal1989efficient} is based on interval labeling. It is one of the initial works to tackle the incremental maintenance of the index in dynamic graphs. Based on 2-HOP labeling, some incremental maintenance methods are proposed in \cite{bramandia2009incremental,schenkel2004hopi,schenkel2005efficient}. However, they do not support efficient delete operations. A recent data structure for labeling, known as DBL \cite{lyu2021dbl}, does support undirected and directed graphs. It only allows edge insertions in graphs, and the construction of DBL is time-consuming due to the need of performing a BFS on the corresponding connected components.  

\stitle{Maintaining 2-edge Connectivity.} In 1991, Fredrickson \cite{frederickson1997ambivalent} introduces a data structure called topology trees for the fully dynamic 2-edge connectivity problem, achieving a worst-case update time of $O(\sqrt{m})$. In 1992, Eppstein et al. \cite{eppstein1997sparsification} improve the update time to $O(\sqrt{n})$ using the sparsification technique. Henzinger et al. \cite{henzinger1997fully} present a fully dynamic 2-edge connectivity algorithm with an amortized expected update time of $O(\log^{5}n)$. Holm et al. \cite{DBLP:journals/jacm/HolmLT01} subsequently propose a deterministic algorithm with $O(\log^{4}n)$ update time for fully dynamic 2-edge connectivity, which remains the state of the art. Recently, some methods are proposed for dynamic 2-edge connectivity in directed graphs \cite{georgiadis2016incremental,georgiadis2025faster}. However, they are partially dynamic methods and do not support fully dynamic updates.

\stitle{Other types of graphs.}
There are many related studies on other types of graphs related to connectivity queries. In temporal graphs, the span-reachability query aims to answer the reachability of any time window \cite{wen2020efficiently,DBLP:journals/vldb/WenYZQCZ22}. Qiao et al. proposed a reachability algorithm on weighted graphs \cite{qiao2013computing}. Reachability in distributed systems is studied in \cite{DBLP:conf/icde/ZhangLQZWCL22}. Label constrained reachability query is to judge whether there is a path between two points that only contains a subset of given labels, which is studied in \cite{peng2020answering,zeng2022distributed}.

\vspace{-1.5em}
\section{Conclusion}
\label{sec:conclu}

In this paper, we propose a new data index for solving connectivity queries in full dynamic graphs. We streamline the data structure of the state-of-the-art algorithm and reduce the time complexity of both insertion and deletion operations. We propose a new strategy to search for replacement edge in edge deletion. Furthermore, we propose a new approach that combines the advantages of spanning tree and disjoint-set tree. We further extend our connectivity maintenance algorithms to maintain 2-edge connectivity between vertices. Our final algorithms achieve the constant query time complexity and also significantly improve the theoretical running time in both edge insertion and edge deletion. Our performance studies on real large datasets show considerable improvement in our algorithms.

\vspace{-2em}

%\begin{acknowledgements}
%\end{acknowledgements}


% \begin{thebibliography}{}
% %
% % and use \bibitem to create references. Consult the Instructions
% % for authors for reference list style.
% % Format for Journal Reference
% Author, Article title, Journal, Volume, page numbers (year)
% Author, Book title, page numbers. Publisher, place (year)
% \end{thebibliography}
%\clearpage
%\begin{thebibliography}{}
%\bibliographystyle{spbasic} 
\bibliographystyle{spmpsci} 
\bibliography{ref.bib}

@String{JACM = "J. ACM" }

@String{Computing = "Computing" }

@String{Computer = "{IEEE} Computer" }

@String{Springer = "Springer-Verlag" }

@article{DBLP:journals/pvldb/ChenLHB22,
  author    = {Qing Chen and
               Oded Lachish and
               Sven Helmer and
               Michael H. B{\"{o}}hlen},
  title     = {Dynamic Spanning Trees for Connectivity Queries on Fully-dynamic Undirected Graphs},
  journal   = {Proc. {VLDB} Endow.},
  volume    = {15},
  number    = {11},
  pages     = {3263--3276},
  year      = {2022}
}

@article{DBLP:journals/jacm/TarjanL84,
  author    = {Robert Endre Tarjan and
               Jan van Leeuwen},
  title     = {Worst-case Analysis of Set Union Algorithms},
  journal   = {J. {ACM}},
  volume    = {31},
  number    = {2},
  pages     = {245--281},
  year      = {1984}
}

@article{DBLP:journals/biodatamining/PavlopoulosSMSKASB11,
  author    = {Georgios A. Pavlopoulos and
               Maria Secrier and
               Charalampos N. Moschopoulos and
               Theodoros G. Soldatos and
               Sophia Kossida and
               Jan Aerts and
               Reinhard Schneider and
               Pantelis G. Bagos},
  title     = {Using graph theory to analyze biological networks},
  journal   = {BioData Min.},
  volume    = {4},
  pages     = {10},
  year      = {2011}
}

@article{agrawal1989efficient,
  title={Efficient management of transitive relationships in large data and knowledge bases},
  author={Agrawal, Rakesh and Borgida, Alexander and Jagadish, Hosagrahar Visvesvaraya},
  journal={ACM SIGMOD Record},
  volume={18},
  number={2},
  pages={253--262},
  year={1989},
  publisher={ACM New York, NY, USA}
}

@article{bramandia2009incremental,
  title={Incremental maintenance of 2-hop labeling of large graphs},
  author={Bramandia, Ramadhana and Choi, Byron and Ng, Wee Keong},
  journal={IEEE Transactions on Knowledge and Data Engineering},
  volume={22},
  number={5},
  pages={682--698},
  year={2009},
  publisher={IEEE}
}

@inproceedings{schenkel2004hopi,
  title={HOPI: An efficient connection index for complex XML document collections},
  author={Schenkel, Ralf and Theobald, Anja and Weikum, Gerhard},
  booktitle={Advances in Database Technology-EDBT 2004: 9th International Conference on Extending Database Technology, Heraklion, Crete, Greece, March 14-18, 2004 9},
  pages={237--255},
  year={2004},
  organization={Springer}
}

@inproceedings{schenkel2005efficient,
  title={Efficient creation and incremental maintenance of the hopi index for complex xml document collections},
  author={Schenkel, Ralf and Theobald, Anja and Weikum, Gerhard},
  booktitle={21st International Conference on Data Engineering (ICDE'05)},
  pages={360--371},
  year={2005},
  organization={IEEE}
}

@article{tarjan1975efficiency,
	title={Efficiency of a good but not linear set union algorithm},
	author={Tarjan, Robert Endre},
	journal={Journal of the ACM (JACM)},
	volume={22},
	number={2},
	pages={215--225},
	year={1975},
	publisher={ACM New York, NY, USA}
}

@article{shiloach1981line,
	title={An on-line edge-deletion problem},
	author={Shiloach, Yossi and Even, Shimon},
	journal={Journal of the ACM (JACM)},
	volume={28},
	number={1},
	pages={1--4},
	year={1981},
	publisher={ACM New York, NY, USA}
}

@article{eppstein1997sparsification,
	title={Sparsification—a technique for speeding up dynamic graph algorithms},
	author={Eppstein, David and Galil, Zvi and Italiano, Giuseppe F and Nissenzweig, Amnon},
	journal={Journal of the ACM (JACM)},
	volume={44},
	number={5},
	pages={669--696},
	year={1997},
	publisher={ACM New York, NY, USA}
}

@inproceedings{tarjan1984finding,
	title={Finding biconnected componemts and computing tree functions in logarithmic parallel time},
	author={Tarjan, Robert Endre and Vishkin, Uzi},
	booktitle={25th Annual Symposium onFoundations of Computer Science, 1984.},
	pages={12--20},
	year={1984},
	organization={IEEE}
}

@inproceedings{henzinger1995randomized,
	title={Randomized dynamic graph algorithms with polylogarithmic time per operation},
	author={Henzinger, Monika Rauch and King, Valerie},
	booktitle={Proceedings of the twenty-seventh annual ACM symposium on Theory of computing},
	pages={519--527},
	year={1995}
}

@article{henzinger1999randomized,
	title={Randomized fully dynamic graph algorithms with polylogarithmic time per operation},
	author={Henzinger, Monika R and King, Valerie},
	journal={Journal of the ACM (JACM)},
	volume={46},
	number={4},
	pages={502--516},
	year={1999},
	publisher={ACM New York, NY, USA}
}

@inproceedings{cheng2013tf,
	title={Tf-label: a topological-folding labeling scheme for reachability querying in a large graph},
	author={Cheng, James and Huang, Silu and Wu, Huanhuan and Fu, Ada Wai-Chee},
	booktitle={Proceedings of the 2013 ACM SIGMOD International Conference on Management of Data},
	pages={193--204},
	year={2013}
}

@inproceedings{jin20093,
	title={3-hop: a high-compression indexing scheme for reachability query},
	author={Jin, Ruoming and Xiang, Yang and Ruan, Ning and Fuhry, David},
	booktitle={Proceedings of the 2009 ACM SIGMOD International Conference on Management of data},
	pages={813--826},
	year={2009}
}

@article{wei2014reachability,
	title={Reachability querying: An independent permutation labeling approach},
	author={Wei, Hao and Yu, Jeffrey Xu and Lu, Can and Jin, Ruoming},
	journal={Proceedings of the VLDB Endowment},
	volume={7},
	number={12},
	pages={1191--1202},
	year={2014},
	publisher={VLDB Endowment}
}

@inproceedings{zhu2014reachability,
	title={Reachability queries on large dynamic graphs: a total order approach},
	author={Zhu, Andy Diwen and Lin, Wenqing and Wang, Sibo and Xiao, Xiaokui},
	booktitle={Proceedings of the 2014 ACM SIGMOD international conference on Management of data},
	pages={1323--1334},
	year={2014}
}

@inproceedings{zeng2022distributed,
  title={Distributed Set Label-Constrained Reachability Queries over Billion-Scale Graphs},
  author={Zeng, Yuanyuan and Yang, Wangdong and Zhou, Xu and Xiao, Guoqing and Gao, Yunjun and Li, Kenli},
  booktitle={2022 IEEE 38th International Conference on Data Engineering (ICDE)},
  pages={1969--1981},
  year={2022},
  organization={IEEE}
}

@article{peng2020answering,
  title={Answering billion-scale label-constrained reachability queries within microsecond},
  author={Peng, You and Zhang, Ying and Lin, Xuemin and Qin, Lu and Zhang, Wenjie},
  journal={Proceedings of the VLDB Endowment},
  volume={13},
  number={6},
  pages={812--825},
  year={2020},
  publisher={VLDB Endowment}
}

@inproceedings{wen2020efficiently,
  title={Efficiently answering span-reachability queries in large temporal graphs},
  author={Wen, Dong and Huang, Yilun and Zhang, Ying and Qin, Lu and Zhang, Wenjie and Lin, Xuemin},
  booktitle={2020 IEEE 36th International Conference on Data Engineering (ICDE)},
  pages={1153--1164},
  year={2020},
  organization={IEEE}
}

@article{qiao2013computing,
  title={Computing weight constraint reachability in large networks},
  author={Qiao, Miao and Cheng, Hong and Qin, Lu and Yu, Jeffrey Xu and Yu, Philip S and Chang, Lijun},
  journal={The VLDB journal},
  volume={22},
  number={3},
  pages={275--294},
  year={2013},
  publisher={Springer}
}

@inproceedings{lyu2021dbl,
	title={DBL: Efficient Reachability Queries on Dynamic Graphs},
	author={Lyu, Qiuyi and Li, Yuchen and He, Bingsheng and Gong, Bin},
	booktitle={Database Systems for Advanced Applications: 26th International Conference, DASFAA 2021, Taipei, Taiwan, April 11--14, 2021, Proceedings, Part II 26},
	pages={761--777},
	year={2021},
	organization={Springer}
}

@inproceedings{thorup2000near,
  title={Near-optimal fully-dynamic graph connectivity},
  author={Thorup, Mikkel},
  booktitle={Proceedings of the thirty-second annual ACM symposium on Theory of computing},
  pages={343--350},
  year={2000}
}

@inproceedings{DBLP:conf/stoc/HolmLT98,
  author       = {Jacob Holm and
                  Kristian de Lichtenberg and
                  Mikkel Thorup},
  editor       = {Jeffrey Scott Vitter},
  title        = {Poly-Logarithmic Deterministic Fully-Dynamic Algorithms for Connectivity,
                  Minimum Spanning Tree, 2-Edge, and Biconnectivity},
  booktitle    = {Proceedings of the Thirtieth Annual {ACM} Symposium on the Theory
                  of Computing, Dallas, Texas, USA, May 23-26, 1998},
  pages        = {79--89},
  publisher    = {{ACM}},
  year         = {1998},
  url          = {https://doi.org/10.1145/276698.276715},
  doi          = {10.1145/276698.276715},
  bibsource    = {dblp computer science bibliography, https://dblp.org}
}

@article{DBLP:journals/jacm/HolmLT01,
  author       = {Jacob Holm and
                  Kristian de Lichtenberg and
                  Mikkel Thorup},
  title        = {Poly-logarithmic deterministic fully-dynamic algorithms for connectivity,
                  minimum spanning tree, 2-edge, and biconnectivity},
  journal      = {J. {ACM}},
  volume       = {48},
  number       = {4},
  pages        = {723--760},
  year         = {2001},
  url          = {https://doi.org/10.1145/502090.502095},
  doi          = {10.1145/502090.502095},
  bibsource    = {dblp computer science bibliography, https://dblp.org}
}

@article{iyer2001experimental,
  title={An experimental study of polylogarithmic, fully dynamic, connectivity algorithms},
  author={Iyer, Raj and Karger, David and Rahul, Hariharan and Thorup, Mikkel},
  journal={Journal of Experimental Algorithmics (JEA)},
  volume={6},
  pages={4--es},
  year={2001},
  publisher={ACM New York, NY, USA}
}

@article{huang2023fully,
  title={Fully Dynamic Connectivity in $ O(\log {n} (\log{log {n}})^{2}) $ Amortized Expected Time},
  author={Huang, Shang-En and Huang, Dawei and Kopelowitz, Tsvi and Pettie, Seth and Thorup, Mikkel},
  journal={TheoretiCS},
  volume={2},
  year={2023},
  publisher={Episciences. org}
}

@article{mao2021digital,
  title={Digital contact tracing based on a graph database algorithm for emergency management during the COVID-19 epidemic: Case study},
  author={Mao, Zijun and Yao, Hong and Zou, Qi and Zhang, Weiting and Dong, Ying and others},
  journal={JMIR mHealth and uHealth},
  volume={9},
  number={1},
  pages={e26836},
  year={2021},
  publisher={JMIR Publications Inc., Toronto, Canada}
}

@article{mirza2003studying,
  title={Studying recommendation algorithms by graph analysis},
  author={Mirza, Batul J and Keller, Benjamin J and Ramakrishnan, Naren},
  journal={Journal of intelligent information systems},
  volume={20},
  pages={131--160},
  year={2003},
  publisher={Springer}
}

@inproceedings{maesa2016uncovering,
  title={Uncovering the bitcoin blockchain: an analysis of the full users graph},
  author={Maesa, Damiano Di Francesco and Marino, Andrea and Ricci, Laura},
  booktitle={2016 IEEE international conference on data science and advanced analytics (DSAA)},
  pages={537--546},
  year={2016},
  organization={IEEE}
}

@inproceedings{daniel2020gis,
  title={Gis based road connectivity evaluation using graph theory},
  author={Daniel, Cynthia Baby and Saravanan, S and Mathew, Samson},
  booktitle={Transportation Research: Proceedings of CTRG 2017},
  pages={213--226},
  year={2020},
  organization={Springer}
}

@article{DBLP:journals/pacmmod/SongWXQZ024,
  author       = {Jingyi Song and
                  Dong Wen and
                  Lantian Xu and
                  Lu Qin and
                  Wenjie Zhang and
                  Xuemin Lin},
  title        = {On Querying Historical Connectivity in Temporal Graphs},
  journal      = {Proc. {ACM} Manag. Data},
  volume       = {2},
  number       = {3},
  pages        = {157},
  year         = {2024}
}

@article{xu2026querying,
  title={On Querying Historical Connectivity in Large-scale Temporal Graphs: L. Xu et al.},
  author={Xu, Lantian and Wen, Dong and Qin, Lu and Zhang, Wenjie and Wang, Xubo and Lin, Xuemin},
  journal={The VLDB Journal},
  volume={35},
  number={1},
  pages={2},
  year={2026},
  publisher={Springer}
}

@article{DBLP:journals/vldb/WenYZQCZ22,
  author       = {Dong Wen and
                  Bohua Yang and
                  Ying Zhang and
                  Lu Qin and
                  Dawei Cheng and
                  Wenjie Zhang},
  title        = {Span-reachability querying in large temporal graphs},
  journal      = {{VLDB} J.},
  volume       = {31},
  number       = {4},
  pages        = {629--647},
  year         = {2022}
}

@inproceedings{DBLP:conf/icde/ZhangLQZWCL22,
  author       = {Junhua Zhang and
                  Wentao Li and
                  Lu Qin and
                  Ying Zhang and
                  Dong Wen and
                  Lizhen Cui and
                  Xuemin Lin},
  title        = {Reachability Labeling for Distributed Graphs},
  booktitle    = {38th {IEEE} International Conference on Data Engineering, {ICDE} 2022,
                  Kuala Lumpur, Malaysia, May 9-12, 2022},
  pages        = {686--698},
  publisher    = {{IEEE}},
  year         = {2022}
}

@article{frederickson1997ambivalent,
  title={Ambivalent data structures for dynamic 2-edge-connectivity and k smallest spanning trees},
  author={Frederickson, Greg N},
  journal={SIAM Journal on Computing},
  volume={26},
  number={2},
  pages={484--538},
  year={1997},
  publisher={SIAM}
}

@article{henzinger1997fully,
  title={Fully dynamic 2-edge connectivity algorithm in polylogarithmic time per operation},
  author={Henzinger, Monika Rauch and King, Valerie},
  journal={SRC Technical Note},
  volume={4},
  pages={4},
  year={1997}
}

@inproceedings{georgiadis2016incremental,
  title={Incremental 2-Edge-Connectivity in Directed Graphs},
  author={Georgiadis, Loukas and Italiano, Giuseppe F and Parotsidis, Nikos},
  booktitle={43rd International Colloquium on Automata, Languages, and Programming (ICALP 2016)},
  volume={55},
  pages={49},
  year={2016},
  organization={Schloss Dagstuhl--Leibniz-Zentrum fuer Informatik}
}

@inproceedings{georgiadis2025faster,
  title={Faster Dynamic 2-Edge Connectivity in Directed Graphs},
  author={Georgiadis, Loukas and Giannis, Konstantinos and Italiano, Giuseppe F},
  booktitle={33rd Annual European Symposium on Algorithms (ESA 2025)},
  pages={26--1},
  year={2025},
  organization={Schloss Dagstuhl--Leibniz-Zentrum f{\"u}r Informatik}
}

@article{krekovic2025reducing,
  title={Reducing communication overhead in the IoT-edge-cloud continuum: A survey on protocols and data reduction strategies},
  author={Krekovi{\'c}, Dora and Krivi{\'c}, Petar and {\v{Z}}arko, Ivana Podnar and Ku{\v{s}}ek, Mario and Le-Phuoc, Danh},
  journal={Internet of things},
  pages={101553},
  year={2025},
  publisher={Elsevier}
}

@article{chaudhuri2022network,
  title={Network approach to understand biological systems: From single to multilayer networks},
  author={Chaudhuri, Sayantoni and Srivastava, Ashutosh},
  journal={Journal of Biosciences},
  volume={47},
  number={4},
  pages={55},
  year={2022},
  publisher={Springer}
}

%\end{thebibliography}
\end{document}